\documentclass{article}

\usepackage[preprint]{neurips_2026}

\usepackage[utf8]{inputenc} 
\usepackage[T1]{fontenc}    
\usepackage{hyperref}       
\usepackage{url}            
\usepackage{booktabs}       
\usepackage{amsfonts}       
\usepackage{nicefrac}       
\usepackage{microtype}      
\usepackage{xcolor}         

\usepackage{graphicx}
\usepackage{subcaption}
\usepackage{booktabs} 
\usepackage{hyphenat} 

\usepackage{amsmath}
\usepackage{amssymb}
\usepackage{mathtools}
\usepackage{amsthm}
\usepackage{bbm}
\usepackage{float}
\usepackage{placeins}
\usepackage{enumitem}
\usepackage{booktabs}
\usepackage{longtable}
\usepackage{array}
\usepackage{bm}         
\usepackage{xspace}     
\usepackage{algorithm}  
\usepackage{algorithmic}

\usepackage[capitalize,noabbrev]{cleveref}

\providecommand{\pref}[1]{\cref{#1}}

\providecommand{\xb}{\bm{x}}

\providecommand{\zb}{\bm{z}}
\providecommand{\ub}{\bm{u}}

\providecommand{\Lb}{\bm{L}}

\providecommand{\epsilonb}{\bm{\epsilon}}

\providecommand{\thetab}{\bm{\theta}}
\providecommand{\taub}{\bm{\tau}}

\providecommand{\Hcal}{\mathcal{H}}
\providecommand{\Lcal}{\mathcal{L}}
\providecommand{\Gcal}{\mathcal{G}}

\providecommand{\RR}{\mathbb{R}}
\providecommand{\EE}{\mathbb{E}}
\providecommand{\PP}{\mathbb{P}}

\newcommand{\methodname}{\textsc{INCAMA}} 

\newcommand{\Tau}{\bm{\tau}}
\DeclareMathOperator*{\TopK}{TopK}

\theoremstyle{plain}
\newtheorem{theorem}{Theorem}[section]
\newtheorem{proposition}[theorem]{Proposition}
\newtheorem{lemma}[theorem]{Lemma}
\newtheorem{corollary}[theorem]{Corollary}
\theoremstyle{definition}

\newtheorem{assum}[theorem]{Assumption}
\theoremstyle{remark}

\title{Latent-Space Causal Discovery from Indirect Neuroimaging Observations}

\author{%
  Sangyoon Bae \\
  Interdisciplinary Program in Artificial Intelligence \\
  Seoul National University \\
  Seoul, South Korea, 08826 \\
  \texttt{stellasybae@snu.ac.kr} \\
  \And
  Miruna Oprescu \\
  Computer Science \\
  Cornell Tech, Cornell University \\
  New York, NY 10044 \\
  \texttt{amo78@cornell.edu}
  \And
  David Keetae Park \\
  Computational Science Initiative \\
  Brookhaven National Laboratory \\
  Shirley, New York, United States, 11967 \\
  \texttt{dpark1@bnl.gov} \\
  \And
  Shinjae Yoo \\
  Computational Science Initiative \\
  Brookhaven National Laboratory \\
  Shirley, New York, United States, 11967 \\
  \texttt{sjyoo@bnl.gov} \\
  \And
  Jiook Cha \\
  Department of Psychology \\
  Seoul National University \\
  Seoul, South Korea, 08826 \\
  \texttt{connectome@snu.ac.kr} \\
}

\begin{document}

\maketitle

\begin{abstract}
Neuroimaging does not observe causal variables directly: hemodynamics and volume conduction distort signals so that statistical dependence need not reflect latent neural influence.
Before estimating graphs, one must specify \emph{under what assumptions} delayed directed structure can be studied from such indirect observations.
We formalize a conditional setting---recoverable inversion under modality physics together with nonstationary latent dynamics---and derive an inversion-error propagation bound under explicit assumptions.
Building on this framing, we propose \textsc{\methodname} (INdirect CAusal MAmba): physics-aware inversion coupled with a delay-aware Mamba encoder that uses mechanism shifts as informative variation for directed graph scoring.
We use controlled simulations for quantitative validation and HCP motor-task fMRI as a zero-shot external transfer check based on anatomical and task-network consistency.
Across TVB simulations, \textsc{\methodname} improves directed-structure recovery by 2--3$\times$ in F1 over observation-space and two-stage baselines, and on HCP motor-task fMRI it produces sparse directed estimates concentrated in canonical visuo-motor pathways.
\end{abstract}

\section{Introduction}\label{sec:intro}
EEG and fMRI provide large-scale access to human brain activity, but they observe neural dynamics only indirectly through modality-specific forward operators such as hemodynamic filtering and sensor mixing~\cite{friston2003dynamic,schoffelen2009source,peters2017elements}.
This makes causal discovery from neuroimaging different from unconditional recovery of a latent directed graph: under indirect observation, statistical dependencies need not align with neural causal influence.
We therefore ask what can be supported \emph{under explicit assumptions} about inversion quality and latent dynamics.

\textbf{Assumptions we state upfront.}
Our analysis and pipeline rely on hypotheses we treat as falsifiable rather than granted facts:
\textit{(i)}~recoverable latent trajectories under regularized inversion (\pref{assum:invertibility}), which is strong for ill-posed inverse problems;
\textit{(ii)}~cortical ROI sources without full subcortex (\pref{assum:causal_sufficiency}), so edges are effective cortical connectivity conditioned on subcortical mixing;
\textit{(iii)}~simulator-grounded training because causal adjacency is unobservable in humans.
We return to consequences in Section~\ref{sec:conclusion}.

Existing methods mostly split along scalability versus measurement realism: DCM couples neural and observation models but scales poorly~\cite{friston2003dynamic}; regression-DCM trades fidelity for speed~\cite{frassle2021regression}; Granger/VAR and neural causal models applied in observation space remain sensitive to hemodynamics and volume conduction~\cite{granger1969investigating,lutkepohl2005new,tank2021neural}.

We propose treating directed connectivity as \emph{latent}: physics-aware inversion followed by delay-aware discovery with a selective state-space backbone (Mamba)~\cite{gu2024mamba}, using nonstationarity as informative variation~\cite{huang2019causal,scholkopf2021toward}.
Training uses modality-specific priors (DeepSIF-style EEG inversion~\cite{sun2022deep}; ROI-HRF-aware fMRI deconvolution~\cite{friston2003dynamic,handwerker2004variation}).
Our contribution is to connect this practical architecture to the conditional question above: when forward physics can be inverted well enough, what directed structure is identifiable, how do inversion errors affect graph recovery, and what evidence remains appropriate when human ground truth is unavailable?
The theory is therefore a reduction rather than an unconditional identifiability claim: it clarifies how inversion quality constrains recoverable directed structure under indirect observation.
Concretely, we make three contributions:
\begin{itemize}[leftmargin=1.75em, itemsep=0.15em, topsep=0pt, partopsep=0pt]
    \item \textbf{\textsc{\methodname}} couples inversion and latent causal discovery end-to-end for indirect neuroimaging.
    \item \textbf{Theory:} a conditional identifiability reduction under recoverable inversion and nonstationarity, plus an inversion-error propagation bound (\pref{sec:theory}).
    \item \textbf{Evidence:} biophysical simulations with ground truth; real-data plausibility checks under the standard indirect-validation paradigm (Section~\ref{sec:real_fmri}); robustness analyses (Appendix~\ref{app:robust_tables}).
\end{itemize}


\section{Background and Related Work}\label{sec:background}

For a broader survey, see Appendix Section~\ref{app:extended-lit}.
Prior work separates two issues that are coupled in our setting: measurement physics and directed-structure identification.
DCM jointly models neural dynamics and hemodynamics but scales poorly~\cite{friston2003dynamic}, while regression-DCM and observation-space GC/VAR improve scalability at the cost of either reduced biophysical fidelity or sensitivity to mixing and hemodynamic filtering~\cite{frassle2021regression,granger1969investigating,lutkepohl2005new,tank2021neural,shimizu2006linear}.

On the causal-discovery side, nonstationarity provides a route to identifiability because mechanism changes can reveal directed structure~\cite{huang2019causal,peters2017elements,scholkopf2021toward,zhang2017causal,huang2020causal,pfister2019invariant,saggioro2020reconstructing,gao2023causal}, but this literature typically assumes \emph{direct} access to the causal variables.
State-space models enter precisely at this point: they are a natural language for time-varying latent dynamics, underpin classical neurodynamic models such as DCM, and motivate modern selective SSM layers such as Mamba for tracking mechanism shifts~\cite{friston2003dynamic,valdes2011effective,gu2024mamba}.
CausalMamba~\cite{bae2025causalmamba} uses sequence models for causal discovery, but it is not designed for modality physics or whole-brain indirect neuroimaging; our contribution is to make nonstationarity-based latent discovery usable after physics-aware inversion.

\section{Problem Setting: Causal Discovery from Indirect Human Brain Measurements}
\label{sec:problem_setting}
\label{app:problem_formalism}

We treat whole-brain activity as latent ROI/source signals $\zb_t\in\RR^N$ observed indirectly through modality physics.
Directed interactions are encoded by a delayed graph $\Gcal=(V,E,\taub)$ with $V=\{1,\dots,N\}$, edges $E$, and per-edge delays $\taub_{ij}\in\{1,\dots,D\}$.
Define $\textrm{Pa}(j)=\{i:(i\to j)\in E\}$.
The latent process follows a (possibly nonlinear) delayed dynamical SCM,
\vspace{-0.15em}
\begin{equation}
  z_{j,t}
  \;=\;
  f_{j,t}\!\left(\{ z_{i,t-\taub_{ij}} : i\in\textrm{Pa}(j)\},\,\ub_t\right)
  \;+\; \eta_{j,t},
  \quad j\in V,
  \label{eq:delayed_scm}
\end{equation}
\vspace{-0.35em}

\noindent
where $\ub_t$ denotes optional exogenous inputs (set $\ub_t\equiv 0$ for resting-state), $\eta_{j,t}$ are disturbances, and nonstationarity may enter through time-varying mechanisms $f_{j,t}$ (e.g., regime-indexed parameters).

Observations $\xb_t\in\RR^M$ arise from a modality-specific forward operator with finite memory,
\vspace{-0.15em}
\begin{equation}
  \xb_t \;=\; \Hcal_{\psi} \big(\zb_{t-L:t}\big) + \epsilonb_t,
  \label{eq:measurement}
\end{equation}
\vspace{-0.35em}

\noindent
where $\zb_{t-L:t}:=(\zb_{t-L},\dots,\zb_{t})$, $L\ge 0$, and $\psi$ collects nuisance measurement parameters.
A common fMRI specialization is HRF convolution,
$\xb_t \approx \sum_{\ell=0}^L h_{\psi,\ell}\,\zb_{t-\ell} + \epsilonb_t$;
EEG is often approximated as instantaneous leadfield mixing $\xb_t \approx \Lb_{\psi}\,\zb_t + \epsilonb_t$, i.e., \eqref{eq:measurement} with $L{=}0$.
Both modalities can induce statistical dependencies misaligned with latent causal influences unless inversion is modeled.

Our causal estimand is $\Gcal$.
Given $\xb_{1:T}$, we learn a map
\vspace{-0.15em}
\begin{equation}
  \Phi:\ \xb_{1:T}\ \mapsto\ (\widehat{\zb}_{1:T},\ \widehat{\Gcal}),
  \label{eq:phi_objective}
\end{equation}
\vspace{-0.35em}

\noindent
that (i) performs \emph{regularized} inversion tied to $\Hcal_\psi$ (rather than an unconstrained embedding) and (ii) infers sparse directed delays in latent space using nonstationary structure in $\widehat{\zb}_{1:T}$.
Figure~\ref{fig:overall_framework} overviews the end-to-end pipeline that instantiates \eqref{eq:phi_objective}.

\begin{figure*}[!t]
    \centering
    \includegraphics[width=\textwidth]{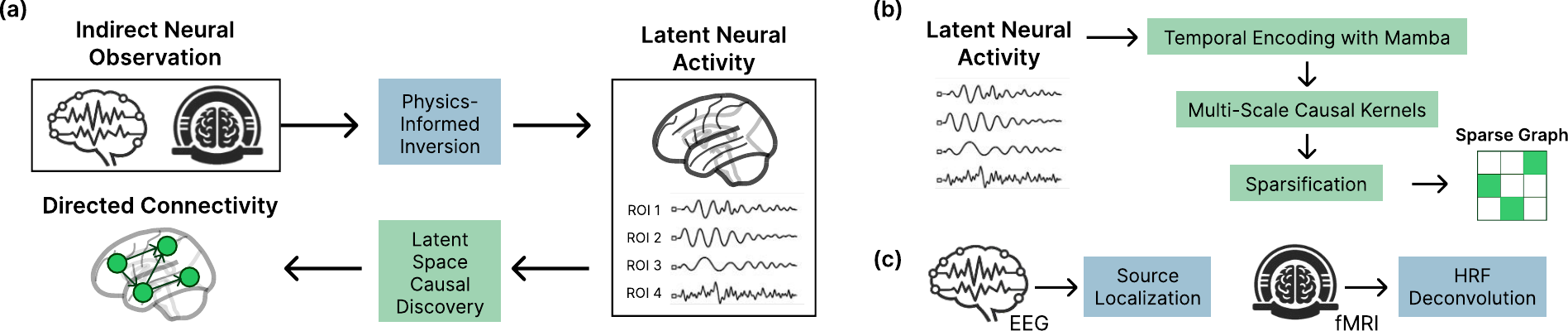}
    \caption{\texorpdfstring{\textbf{\textsc{\methodname} overview.} Inversion (HRF/leadfield) $\rightarrow$ latent Mamba encoding $\rightarrow$ delay-aware sparse graph. Training on simulation; transfer to HCP fMRI in Section~\ref{sec:real_fmri}.}{INCAMA overview: inversion, latent Mamba encoding, delay-aware sparse graph.}}
    \label{fig:overall_framework}
\end{figure*}

\section{Identifiability under Indirect Observation and Nonstationarity}
\label{sec:theory}

We ask when the delayed graph $\Gcal$ governing $\zb_{1:T}$ is identifiable from $\xb_t = \Hcal_\psi(\zb_{t-L:t}) + \epsilonb_t$ (unique $(E,\taub)$ from $P(\xb_{1:T})$ within the model class).
Ill-conditioning of $\Hcal_\psi$ and directional ambiguity under stationarity split the question into (a) recoverable inversion and (b) latent identifiability via nonstationary ``soft interventions'' (ICM; \citet{peters2017elements, huang2019causal}).
We state assumptions and results next (proofs in \pref{app:theory_proofs}).

\subsection{Assumptions}

We record the structural hypotheses used below; \emph{graphical/dynamical regularity} (Assumptions~\ref{assum:no_instantaneous}--\ref{assum:causal_sufficiency}), \emph{identifying variation} (Assumption~\ref{assum:nonstationary_softint}), \emph{recoverable inversion} (Assumption~\ref{assum:invertibility}), and \emph{latent identifiability class} (Assumption~\ref{assum:latent_identifiable_class}).

\begin{assum}[No instantaneous effects (delayed feedback)]
\label{assum:no_instantaneous}
All directed interactions have strictly positive delay: $\taub_{ij} \ge 1$ for all $(i \to j)\in E$.
\end{assum}

\begin{assum}[Latent causal sufficiency]
\label{assum:causal_sufficiency}
The exogenous noises $\{\eta_{j,t}\}_{j\in V}$ are jointly independent across $j$ (conditional on the past), so there are no unmodeled latent confounders among the latent ROIs/sources represented by $\zb_t$.
\end{assum}

\noindent
Assumption~\ref{assum:no_instantaneous} excludes instantaneous effects so causal influence propagates forward in time; Assumption~\ref{assum:causal_sufficiency} is the usual ``no hidden confounding among modeled units'' condition (in practice, a cortical ROI choice with subcortical effects absorbed as mixing).

\begin{assum}[Nonstationary mechanisms as soft interventions]
\label{assum:nonstationary_softint}
There exists an environment index $e(t)\in\{1,\dots,E\}$ such that
$f_{j,t}(\cdot)=f_j(\cdot;\thetab_{j}^{e(t)})$.
Moreover, the mechanism shifts $\{\thetab_j^{e}\}_e$ vary independently across modules $j$ (ICM)
and are sufficiently informative (non-degenerate) to distinguish candidate graphs within the model class.
\end{assum}

\noindent
Assumption~\ref{assum:nonstationary_softint} formalizes mechanism shifts as informative soft interventions \citep{huang2019causal,peters2017elements}.

\begin{assum}[Recoverable latent state under indirect observation]
\label{assum:invertibility}
There exists a (regularized) inversion procedure $g$ such that
$\widehat{\zb}_{1:T} = g(\xb_{1:T})$ and
\vspace{-0.5em}
\begin{equation}
  \frac{1}{T}\sum_{t=1}^T \| \widehat{\zb}_t - \zb_t \|_2^2 \xrightarrow[T\to\infty]{} 0
  \quad \text{(in probability)},
  \label{eq:latent_recovery}
\end{equation}
\vspace{-0.35em}

\noindent
under the measurement model, up to known equivalences (e.g., per-node scaling conventions).
Empirical support under forward-model perturbations is provided in Appendix~\ref{app:robust_tables}.
\end{assum}

\begin{assum}[\texorpdfstring{Latent identifiability via \citet{huang2019causal}}{Latent identifiability (Huang et al., 2019)}]
\label{assum:latent_identifiable_class}
Assume the latent process $\zb_{1:T}$ lies in the identifiable nonstationary SEM/SSM class of
\citet{huang2019causal}, adapted to our delayed-parent setting (parents of $z_{j,t}$ may include $\{z_{i,t-\taub_{ij}}\}$ with $\taub_{ij}\in\{1,\dots,D\}$),
and that Assumption~\ref{assum:nonstationary_softint} holds.
Then, within this model class, the delayed graph $\Gcal=(V,E,\taub)$ is identifiable from the joint distribution of $\zb_{1:T}$ (as $T\to\infty$).
\end{assum}

\subsection{Main Results (proofs deferred)}

Under \pref{assum:latent_identifiable_class}, the delayed graph is identifiable from $P(\zb_{1:T})$ by \citet{huang2019causal} (with the above delayed-parent adaptation). Our novel contributions are: (i) a reduction showing that consistent inversion (\pref{assum:invertibility}) transfers this guarantee to indirect
observations $\xb_{1:T}$, and (ii) an error bound quantifying how inversion error impacts downstream graph recovery.

\begin{theorem}[Identifiability from indirect observations via inversion]
\label{thm:indirect_identifiability}
Consider the model class where (i) $\zb_{1:T}$ is generated by a delayed dynamical SCM satisfying
\pref{assum:no_instantaneous}--\pref{assum:latent_identifiable_class}, and (ii) observations follow
$\xb_t=\Hcal_\psi(\zb_{t-L:t})+\epsilonb_t$ for which \pref{assum:invertibility} holds.
Then the delayed graph $\Gcal=(V,E,\taub)$ is identifiable from the distribution of $\xb_{1:T}$ (as $T\to\infty$),
up to the known equivalences in \pref{assum:invertibility}.
\end{theorem}


\pref{thm:indirect_identifiability} shows that consistent inversion transfers latent-space identifiability to indirect observations. Next, we quantify how inversion error propagates to graph estimation error.


\textbf{Scores and top-$k$.}
We estimate continuous scores $S\in\mathbb{R}^{|V|\times|V|}$ (optionally with delay channels), summarize them as a weighted directed graph, and retain the top-$k$ strongest entries as a sparse edge set---a pragmatic analogue of hard sparsity that aligns with biophysical expectations of sparse anatomical substrate.
Formal stability/margin discussion is deferred to \pref{app:theory_proofs}.

\begin{proposition}[Inversion error propagation]
\label{prop:error_propagation}
Let $\widehat{\zb}_{1:T}=g(\xb_{1:T})$ and $\widehat{S}=h(\widehat{\zb}_{1:T})$, where $h$ is a continuous score function.
Let $S^\star$ denote the score representation of the true delayed graph $\Gcal$, and let $\Lcal$ be a loss on score objects.
Define the distance
$
d_T(\zb,\zb') := (\frac{1}{T}\sum_{t=1}^T \|\zb_t-\zb'_t\|_2^2)^{1/2}.
$
Assume $h$ is $L$-Lipschitz w.r.t.\ $d_T$ and $\Lcal$:
\[
\Lcal(h(\zb),h(\zb')) \le L\, d_T(\zb,\zb') \qquad \forall\,\zb,\zb'.
\]
Then
\begin{equation}
\EE\!\left[\Lcal(\widehat{S},S^\star)\right]
\;\le\;
\EE\!\left[\Lcal\!\big(h(\zb_{1:T}),S^\star\big)\right]
 +
L\left(\EE\!\left[\frac{1}{T}\sum_{t=1}^T \|\widehat{\zb}_t-\zb_t\|_2^2\right]\right)^{1/2}.
\label{eq:error_propagation_expectation}
\end{equation}
\end{proposition}
In \pref{prop:error_propagation}, the first term is the intrinsic error of the latent graph estimator (even with true $\zb$); the second term quantifies how inversion error degrades graph recovery. 

\begin{corollary}[Consistency of the end-to-end estimator]
\label{corr:consistency}
Assume \pref{assum:invertibility} and that, given access to $\zb_{1:T}$, the score estimator $h$ is consistent under the latent model class in \pref{assum:latent_identifiable_class}:
$\Lcal(h(\zb),S^\star)\xrightarrow[T\to\infty]{}0$ in probability.
Then the end-to-end estimator $\widehat{S}=h(g(\xb))$ is consistent, i.e.
$\Lcal(\widehat{S},S^\star)\xrightarrow[T\to\infty]{}0$.
\end{corollary}

\textbf{Takeaway.}
Nonstationary mechanism shifts yield latent-space identifiability (via \citet{huang2019causal}); if the latent trajectories are consistently recoverable from indirect measurements, the delayed graph is identifiable from observed EEG/fMRI, with stability controlled by the inversion error bound in \pref{prop:error_propagation}.

\section{INCAMA: Indirect Causal Mamba}
\label{sec:method}


We detail \textbf{\textsc{\methodname}} as implemented in Figure~\ref{fig:overall_framework}; pseudocode is in Appendix~\ref{appendix:algorithm}.
The architecture follows the problem factorization emphasized above: a modality-specific Stage~1 inverts the forward physics to obtain $\widehat{\zb}_{1:T}$ that are committed to being neural-scale trajectories (not generic embeddings), and a Stage~2 performs delay-aware, nonstationarity-exploiting causal discovery in latent space with a Mamba-based temporal backbone.
The Mamba choice is therefore not an arbitrary architectural claim: it is a scalable selective-SSM instantiation of the state-space perspective used in nonstationary causal discovery~\citep{huang2019causal,gu2024mamba}, although the identifiability statements in Section~\ref{sec:theory} are conditional on the model class and inversion assumptions rather than on Mamba itself.
Selective SSMs are structurally aligned with the problem because they parameterize latent dynamics with long-range delayed dependencies under nonstationary regime shifts, not only because they scale efficiently.
End-to-end training on simulated data lets the causal objective reshape inversion so that reconstructions are useful for graph identifiability, not only for low reconstruction error in $\xb$.
\textsc{\methodname} thus targets the pair $(\widehat{\zb}_{1:T},\widehat{G})$ jointly.

\noindent
We instantiate the map $\Phi$ from \eqref{eq:phi_objective}; $x_{1:T}$ denotes observed EEG or fMRI and $\widehat G$ the inferred directed graph.
We follow the notation of Sections~\ref{sec:problem_setting}--\ref{sec:theory}.

\subsection{Physics-aware Inversion}
\label{subsec:inversion}

Neuroimaging measurements provide only indirect observations of latent neural activity $z_t$, generated through a known but non-invertible forward process
\begin{equation}
x_t = H_\psi(z_{t-L:t}) + \varepsilon_t ,
\end{equation}
rendering recovery of $z_{1:T}$ from $x_{1:T}$ an ill-posed inverse problem.
To address this, we learn a parameterized inverse model $\widehat z_{1:T} = g_\theta(x_{1:T})$, whose architecture and training objectives are constrained by the known forward operator $H_\psi$, ensuring that the recovered trajectories correspond to physically meaningful neural activity rather than generic latent embeddings.


\textbf{Case 1: Spatial inversion in EEG.}
EEG measurements are approximately instantaneous linear mixtures of neural sources~\citep{baillet2001electromagnetic,mosher1999forward},
\begin{equation}
x_t \approx L_\psi z_t + \varepsilon_t ,
\end{equation}
where $L_\psi \in \mathbb{R}^{M\times R}$ is the leadfield matrix.
We learn an inverse model $\widehat z_{1:T} = g^{\mathrm{EEG}}_\theta(x_{1:T})$ constrained such that $L_\psi \widehat z_t$ reconstructs $x_t$, promoting causal discovery on spatially unmixed latent neural activity.

\textbf{Case 2: Temporal inversion in fMRI.}
fMRI measurements are temporally blurred by the hemodynamic response function (HRF)~\citep{friston2003dynamic,logothetis2008what},
\begin{equation}
x_t \approx \sum_{\ell=0}^{L} h_{\psi,\ell}\, z_{t-\ell} + \varepsilon_t .
\end{equation}
The inverse model jointly estimates latent neural activity and HRF parameters, $(\widehat z_{1:T}, \widehat h_\psi) = g^{\mathrm{fMRI}}_\theta(x_{1:T})$, such that convolving $\widehat z_{1:T}$ with $\widehat h_\psi$ reconstructs the observed BOLD signal.
Allowing region-specific HRFs mitigates hemodynamically induced temporal confounds~\citep{handwerker2004variation,deshpande2010hrfGC,rangaprakash2023hrfconfound}.



\subsection{Latent-Space Causal Discovery}
\label{subsec:causal}

Given reconstructed neural activity, the second stage estimates the directed, delay-aware causal graph $\widehat G$ governing the latent dynamics.

\textbf{Temporal encoding and latent dynamics.}
Each ROI series $\widehat z_{i,1:T}$ passes through a shared Mamba encoder~\citep{gu2024mamba}.
Directed effects use a multi-lag linearized model with short/mid/long delay groups
$\mathcal{G}=\{\text{short},\text{mid},\text{long}\}$,
$
\widehat z_{t+1} =
\sum_{g \in \mathcal{G}} \sum_{\ell \in \mathcal{L}_g}
A^{(g)}_\ell \widehat z_{t-\ell} + \xi_t .
$
Aggregated kernel magnitudes yield causal strengths; top-$k$ pruning yields a sparse graph.

\subsection{Training Objective and End-to-End Optimization}

On simulations with known $A^\star$, we minimize
$\mathcal L_{\mathrm{total}}=
\mathcal L_{\mathrm{graph}}
+
\lambda_{\mathrm{asym}} \mathcal L_{\mathrm{asym}}
+
\lambda_{\mathrm{stab}} \mathcal L_{\mathrm{stab}}$
end-to-end in $\theta$ so gradients from the causal term train $g_\theta$ (Appendix~\ref{appendix:implementation_details}).

\textbf{Transfer to real data.}
After pretraining on biophysical simulations, the learned inference map $\Phi$ is applied zero-shot to real fMRI data without finetuning (Section~\ref{sec:real_fmri}).

\section{Experiments}\label{sec:experiments}
\subsection{Experimental Setup}

\textbf{Simulated data.}
We evaluate INCAMA on large-scale physics-aware simulations inspired by The Virtual Brain (TVB)~\citep{sanz2013virtual}, which generates whole-brain neural activity with known directed connectivity, transmission delays, and controlled nonstationarity.
Realistic EEG and fMRI observations are obtained via biophysically grounded forward models.
The simulators are designed to preserve the factors that matter for graph recovery rather than only matching marginal signal statistics: sparse anatomical substrate, strictly positive conduction delays, time-varying coupling strengths, modality-specific measurement distortion, and realistic sampling/noise.
EEG simulations include leadfield mixing and source leakage, while fMRI simulations include ROI-specific HRFs, high-resolution neural dynamics, HRF convolution, and TR downsampling.
This yields a controlled stress test for indirect causal recovery: the ground-truth graph is known, but the observations retain the main distortions that make human EEG/fMRI difficult.
Simulation details are provided in Appendix~\ref{sec:appendix_EEG_simulator} and Appendix~\ref{sec:appendix_fMRI_simulator} for EEG and fMRI, respectively.

\textbf{Real data.}
As an external plausibility check, we use motor-task fMRI data from the Human Connectome Project (HCP S1200), parcellated into 68 cortical regions using the Desikan--Killiany atlas~\citep{desikan2006automated}.
All models are trained exclusively on simulated data and transferred zero-shot to HCP data, where evaluation focuses on whether high-scoring directed edges align with established motor-network organization.
Because ground-truth directed connectivity is not available in humans, we follow an \emph{indirect evaluation paradigm} common in neuroimaging-based causal discovery: (i) quantitative recovery when the graph is known in simulation; (ii) anatomical and task-network face validity on real data; and (iii) stability under physically plausible forward-model perturbations (Appendix~\ref{app:robust_tables}).
We do not interpret the HCP experiment as ground-truth causal validation.

\textbf{Baselines.}
We compare against standard \emph{two-stage} pipelines that decouple (i) inversion or deconvolution from (ii) linear or learned causal discovery.
For EEG, this includes classical source localization (MNE, sLORETA, DeepSIF) paired with Granger/VAR and with TCN/GNN-based structure learners; for fMRI, FIR/Wiener deconvolution followed by Granger/VAR, rDCM, and neural baselines (TCDF, CNN).
This family isolates the benefit of joint physics-aware inversion and latent causal modeling; CausalMamba is discussed as related sequence-modeling work but is not included as a head-to-head baseline because it is not designed for whole-brain indirect neuroimaging with modality-specific forward physics.
Full names, implementation details, and hyperparameters are in Appendix~\ref{appendix:baselines}.

\textbf{Evaluation metrics.}
On simulations, graph recovery is evaluated using F1 score, normalized Structural Hamming Distance (SHD) and Direction-aware Structural Hamming Distance (dSHD).
All metrics are computed against ground-truth graphs for simulated data; formal definitions are provided in Appendix~\ref{appendix:metrics}.
TVB benchmarks are summarized in Tables~\ref{tab:eeg_results}--\ref{tab:fmri_results}.

\subsection{Results on Simulated Data and Ablation Analysis}\label{app:sim_tables}

\textbf{EEG.}
On simulated EEG (Table~\ref{tab:eeg_results}), classical source-localization-plus-linear-causal pipelines perform near chance under residual volume conduction and source leakage; neural baselines offer modest gains; \textsc{\methodname} substantially improves F1, SHD, and dSHD. Ablations show that omitting source localization causes a pronounced degradation, while exploiting time-varying latent dynamics yields a large observed improvement across the three repeated runs currently reported ($\Delta\text{F1}{=}{+}0.391$; reported descriptively because $n{=}3$).
The failure of the \emph{without physical inversion} ablation is the expected negative control: graph learning directly on mixed sensor observations cannot recover ROI-level directed structure.
The stationary/nonstationary split then tests the second stage rather than the inverse problem, showing that once signals are unmixed, temporally varying mechanisms provide useful orientation signal.

\textbf{fMRI.}
On simulated fMRI (Table~\ref{tab:fmri_results}), classical HRF deconvolution combined with linear estimators remains strongly affected by hemodynamic confounds, while \textsc{\methodname} jointly models HRF variability and delayed interactions in latent space and achieves the strongest agreement with the simulated ground-truth graph.
The \emph{without HRF deconvolution} variant remains competitive but trails the full model, indicating that causal supervision can learn useful structure from BOLD while explicit HRF inversion still improves agreement with the latent graph.
Stationary and nonstationary \textsc{\methodname} variants are practically indistinguishable here ($\Delta\text{F1}{=}{+}0.0008$ across the three repeated runs currently reported); this is theoretically expected---at TR$={}$2~s the HRF smoothing attenuates coupling fluctuations below the sampling resolution---so the gain from explicit nonstationarity emerges precisely in the modality that can resolve the regime shifts.
At matched scale, \textsc{\methodname} reaches F1 $0.610$ at $5{\times}10^{-3}$ GFLOPs vs.\ rDCM $\sim$$10^1$ GFLOPs (F1 $0.258$); see Table~\ref{tab:complexity}.

\textbf{Nonstationarity-aware baseline.}
Against CD-NOD~\citep{huang2020causal} on the matched nonstationary fMRI test set, \textsc{\methodname} attains $13.1\times$ higher F1 and runs $\sim$$188\times$ faster even when CD-NOD is given oracle context labels (Appendix~\ref{app:cdnod}).
This comparison isolates the benefit of doing nonstationary discovery after physics-aware inversion rather than directly on BOLD.

\begin{table*}[t]
\centering
\normalsize
\setlength{\tabcolsep}{6pt}
\setlength{\aboverulesep}{0.5ex}
\setlength{\belowrulesep}{0.5ex}
\renewcommand{\arraystretch}{1.10}
\captionsetup[sub]{skip=3pt,margin=0pt,font=normalsize}
\setlength{\abovecaptionskip}{1pt}
\setlength{\belowcaptionskip}{0pt}

\begin{subtable}{\linewidth}
\centering
\begin{tabular}{@{}lccc@{}}
\toprule
\textbf{Method} & \textbf{F1} $\uparrow$ & \textbf{SHD} $\downarrow$ & \textbf{dSHD} $\downarrow$ \\
\midrule
MNE + Granger & 0.163$\pm$0.000 & 0.530$\pm$0.000 & 0.278$\pm$0.000 \\
MNE + VAR & 0.160$\pm$0.002 & 0.533$\pm$0.002 & 0.280$\pm$0.001 \\
sLORETA + Granger & 0.163$\pm$0.000 & 0.530$\pm$0.000 & 0.278$\pm$0.000 \\
sLORETA + VAR & 0.162$\pm$0.003 & 0.531$\pm$0.003 & 0.279$\pm$0.002 \\
DeepSIF + Granger & 0.176$\pm$0.001 & 0.533$\pm$0.001 & 0.279$\pm$0.000 \\
DeepSIF + TCN + sparse MLP & 0.215$\pm$0.000 & 0.495$\pm$0.000 & 0.254$\pm$0.000 \\
sLORETA + TCN + sparse MLP & 0.174$\pm$0.003 & 0.534$\pm$0.002 & 0.280$\pm$0.001 \\
DeepSIF + GNN + sparse MLP & 0.192$\pm$0.000 & 0.511$\pm$0.001 & 0.265$\pm$0.000 \\
sLORETA + GNN + sparse MLP & 0.193$\pm$0.001 & 0.509$\pm$0.001 & 0.264$\pm$0.001 \\
\midrule
\methodname\ (without physical inversion) & 0.000$\pm$0.000 & 0.102$\pm$0.000 & 0.075$\pm$0.000 \\
\methodname\ (Stationary) & 0.252$\pm$0.034 & 0.148$\pm$0.004 & 0.163$\pm$0.002 \\
\textbf{\methodname\ (Nonstationary)} & \textbf{0.643$\pm$0.012} & \textbf{0.081$\pm$0.010} & \textbf{0.109$\pm$0.009} \\
\bottomrule
\end{tabular}
\caption{\texorpdfstring{TVB-simulated EEG. Structure metrics (MAE ill-defined under nonstationarity); degenerate F1$\approx$0 yields misleadingly low SHD.}{(a) TVB-simulated EEG. MAE ill-defined; degenerate F1 near 0 yields misleadingly low SHD.}}
\label{tab:eeg_results}
\end{subtable}

\vspace{0.35em}

\begin{subtable}{\linewidth}
\centering
\begin{tabular}{@{}lccc@{}}
\toprule
\textbf{Method} & \textbf{F1} $\uparrow$ & \textbf{SHD} $\downarrow$ & \textbf{dSHD} $\downarrow$ \\
\midrule
FIR + Granger & 0.247$\pm$0.000 & 0.400$\pm$0.000 & 0.384$\pm$0.001  \\
FIR + VAR & 0.259$\pm$0.000 & 0.450$\pm$0.000 & 0.460$\pm$0.000  \\
Wiener + Granger & 0.229$\pm$0.000 & 0.352$\pm$0.000 & 0.350$\pm$0.001 \\
Wiener + VAR & 0.257$\pm$0.000 & 0.450$\pm$0.000 & 0.430$\pm$0.000  \\
rDCM & 0.258$\pm$0.006 & 0.450$\pm$0.003 & 0.460$\pm$0.003  \\
Deconv + TCDF & 0.000$\pm$0.000 & 0.174$\pm$0.001 & 0.177$\pm$0.001 \\
EndtoEndCNN & 0.061$\pm$0.073 & 0.534$\pm$0.022 & 0.454$\pm$0.002 \\
\midrule
\methodname\ (without HRF deconvolution) & 0.583$\pm$0.006 & 0.093$\pm$0.002 & 0.093$\pm$0.002 \\
\methodname\ (Stationary) & 0.608$\pm$0.001 & 0.089$\pm$0.000 & 0.093$\pm$0.000 \\
\textbf{\methodname\ (Nonstationary)} & \textbf{0.610$\pm$0.005} & \textbf{0.088$\pm$0.001} & \textbf{0.092$\pm$0.001} \\
\bottomrule
\end{tabular}
\caption{TVB-simulated fMRI; same metrics as (a).}
\label{tab:fmri_results}
\end{subtable}

\vspace{0.2em}
\caption{\texorpdfstring{\textbf{Simulation benchmarks (TVB).} Classical pipelines are strongly affected by mixing/HRF blur; \textsc{\methodname} gives the closest directed-graph agreement with simulator ground truth.}{Simulation benchmarks (TVB). Classical pipelines are affected by mixing/HRF blur; INCAMA gives the closest graph agreement with simulator ground truth.}}
\end{table*}

\subsection{Computational Cost vs.\ Accuracy}\label{app:comp}
\begin{table*}[t]
\centering
\normalsize
\renewcommand{\arraystretch}{1.08}
\begin{tabular}{@{}lccc@{}}
\toprule
\textbf{Method} & \textbf{Complexity} & \textbf{GFLOPs} & \textbf{F1 $\uparrow$} \\
\midrule
rDCM           & $\mathcal{O}(I\times [T R^3 + R^4])$ & $1.0\times 10^{1}$ & 0.258  \\
EndtoEndCNN    & $\mathcal{O}(R T + R^2)$              & $2.0\times 10^{-5}$ & 0.061 \\
Deconv+TCDF    & $\mathcal{O}(R T + R^2)$              & $3.5\times 10^{-4}$ & 0.289 \\
\midrule
\textsc{\methodname} & $\mathcal{O}(R T^2 + R^2 T)$  & $\mathbf{5.0\times 10^{-3}}$ & \textbf{0.610}\\
\bottomrule
\end{tabular}
\caption{\texorpdfstring{Accuracy--cost trade-off on TVB fMRI ($R{=}68$, $T{=}240$; rDCM $I{=}100$). Stage~1+2 inference per subject.}{Accuracy-cost trade-off on TVB fMRI (R=68, T=240; rDCM I=100). Stage 1+2 inference per subject.}}
\label{tab:complexity}
\end{table*}

\subsection{Robustness to Physics Misspecification}\label{sec:robustness}
Simulation-trained pipelines can either memorize simulator physics or learn transferable structure.
These experiments operationalize \pref{thm:indirect_identifiability} by perturbing the forward model while keeping the latent causal structure fixed.
With weights frozen, we stress-test HRF scaling ($\pm20\%$) on TVB fMRI and Frobenius-noise leadfields ($\alpha\in\{0,1,5\}$) on TVB EEG; F1 changes by less than $0.002$ under HRF scaling and drops by at most $5.1\%$ under heavy leadfield noise.
Graph scores therefore move smoothly with bounded drops---consistent with \pref{prop:error_propagation}---rather than collapsing, which argues against a model that merely memorizes nominal simulator physics.
Protocols, tables, and interpretation are in Appendix~\ref{app:robust_tables}.

\subsection{Zero-shot Plausibility Check on Real fMRI}\label{sec:real_fmri}

We apply \textsc{\methodname} zero-shot to HCP S1200 motor-task fMRI ($n{=}1{,}079$ subjects), assessing consistency with the canonical visuo-motor hierarchy \citep{nassi2007specialized,hoshi2004differential}.
\begin{table}[!t]
\centering
\begin{tabular}{lccc}
\toprule
\textbf{Method} & \textbf{Precision} $\uparrow$ & \textbf{Recall} $\uparrow$ & \textbf{F1} $\uparrow$ \\
\midrule
\textbf{\methodname}  & \textbf{0.011$\pm$0.002} & \textbf{0.610$\pm$0.120} & \textbf{0.022$\pm$0.004} \\
rDCM                  & 0.002$\pm$0.002 & 0.140$\pm$0.097 & 0.005$\pm$0.003 \\
FIR + Granger         & 0.003$\pm$0.002 & 0.148$\pm$0.108 & 0.005$\pm$0.004 \\
Random                & 0.002$\pm$0.002 & 0.137$\pm$0.091 & 0.005$\pm$0.003 \\
\bottomrule
\end{tabular}
\setlength{\abovecaptionskip}{3pt}
\setlength{\belowcaptionskip}{0pt}
\vspace*{0.05in}
\caption{\textbf{Overlap with canonical visuo-motor pathways on the HCP S1200 motor task.}
Precision, recall, and F1 are averaged \emph{per subject} across $n{=}1{,}079$ subjects against $12$ canonical visuo-motor edges (mean$\pm$std).
\textsc{\methodname} recovers high recall for canonical visuo-motor edges relative to the strongest baseline (FIR$+$Granger), but absolute precision remains low in the full $68^2$ edge space; we therefore treat this table as an indirect task-network enrichment check rather than ground-truth causal validation.}
\label{tab:causal_recovery}
\vspace*{0.05in}
\end{table}

Fig.~\ref{fig:Causal_discovery_fig} summarizes group-level directed scores; Table~\ref{tab:causal_recovery} reports per-subject overlap with a small set of canonical visuo-motor edges: F1 $0.022$ with low absolute precision, recall $0.610$ vs.\ $0.148$ for FIR$+$Granger; bilateral motor/visual edges (M1$\leftrightarrow$S1, V1$\rightarrow$V2, PM$\rightarrow$M1) appear in $>$$76\%$ of subjects under top-$k$ sparsification (left PM$\rightarrow$M1 is lowest at $\sim$$77\%$; all others exceed $78\%$).
The absolute precision is low, so we interpret this result as task-network enrichment rather than edge-level causal certification.
This low precision is partly structural to the evaluation: among $68\times 67=4{,}556$ possible directed cortical ROI pairs, the canonical reference set labels only $12$ visuo-motor edges as positives, so any plausible task-related edge outside this small list is counted as a false positive.
One likely contributor is cortical-only modeling: subcortical motor structures are not included in the 68-ROI atlas and may act as unmodeled drivers, so HCP edges should be read as effective cortical associations under the stated assumptions.
Observation-space baselines sit near random F1, suggesting that hemodynamic distortion is a plausible bottleneck for direct-BOLD pipelines in this evaluation.
The HCP analysis therefore shows that simulation-trained scores retain task- and anatomy-consistent structure when transferred without finetuning.
Predicted density $14.6\%$ matches $\sim$$15\%$ structural density; edges concentrate in Visual ($40.9\%$) and Somatomotor ($39.6\%$) networks.
Robustness to forward-model perturbation is detailed in Appendix~\ref{app:robust_tables}.

\begin{figure}[t]
    \centering
    \includegraphics[width=0.7\linewidth]{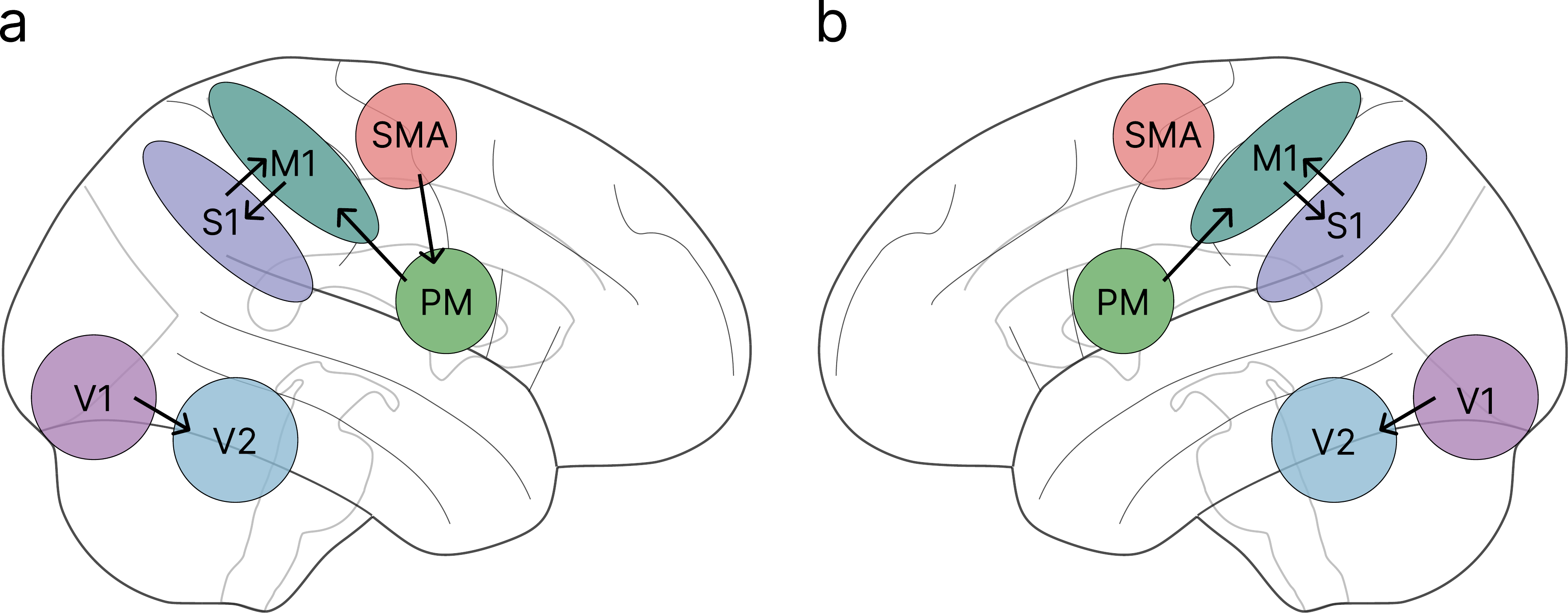}
    \caption{\texorpdfstring{\textbf{Group-averaged visuo-motor pathways} (HCP motor task, $n{=}1{,}079$). Desikan--Killiany ROIs; arrows are cross-subject mean directed estimates.}{Group-averaged visuo-motor pathways (HCP motor task, n=1079). Desikan-Killiany ROIs; arrows are cross-subject mean directed estimates.}}
    \label{fig:Causal_discovery_fig}
\end{figure}

\paragraph{Network-level edge distribution.}\label{app:hcp_network}
Table~\ref{tab:causal_recovery} reports per-subject global metrics; here we summarize the spatial distribution of predicted edges across canonical functional networks for the $1{,}079$-subject HCP cohort during the motor task.
Predicted edges concentrate in task-relevant networks rather than spreading diffusely across the brain:
\begin{itemize}[leftmargin=1.5em,itemsep=0.15em,topsep=0pt]
    \item Visual within-network density: $40.9\%$
    \item Somatomotor within-network density: $39.6\%$
    \item Somatomotor $\leftrightarrow$ Dorsal Attention cross-network density: $32.7\%$ (consistent with attention--motor coupling during voluntary movement)
    \item Overall predicted graph density: $14.6\%$, closely matching structural connectivity density ($\sim$$15\%$) for the Desikan--Killiany atlas
\end{itemize}
The combination of task-network concentration and plausible global density provides convergent evidence that simulation-trained scores transfer to real motor-task fMRI.

\subsection{Linking theory and experiments}\label{sec:theory_exp_link}\label{app:theory_mapping}
Each ingredient in Section~\ref{sec:theory} was mapped to a concrete module \emph{before} running the benchmarks; the experiments are deliberate probes of those assumptions rather than a post-hoc checklist.
Stage~2 implements multi-scale delay-aware kernels and both simulators enforce strictly positive transmission delays, aligning the estimated graph semantics with \pref{assum:no_instantaneous}.
The simulated generators apply time-varying coupling drifts; empirically, EEG benefits from explicit nonstationarity when regime shifts remain temporally resolvable, whereas fMRI shows little stationary/nonstationary split after HRF low-pass filtering.
Stage~1 and the perturbation experiments probe recoverable inversion: under $\pm20\%$ HRF scaling or heavy leadfield Frobenius noise (Appendix~\ref{app:robust_tables}), graph scores degrade smoothly with bounded drops rather than collapsing.
This is the empirical counterpart of the error-propagation view: bounded degradation in the recovered latent trajectory should induce bounded degradation in continuous graph scores, while poor inversion or unmodeled confounding remains the main failure mode.
Continuous scores are finally thresholded with top-$k$ sparsification; the formal margin discussion in Appendix~\ref{app:topk_stability} clarifies that \pref{prop:error_propagation} controls scores, while discrete edge stability additionally requires separation between retained and discarded scores.
This mapping is intended as a compact sanity check, not a separate validation theorem: quantitative claims come from simulated ground truth, while HCP remains an external plausibility check.

\color{black}

\section{Conclusion}\label{sec:conclusion}
\textsc{\methodname} is a neuroimaging causal discovery framework that treats directed connectivity as \emph{latent}: measurement inversion and delayed graph scoring are coupled rather than conflated, and temporal variation is used as identifying signal rather than nuisance.
We formulated sufficient conditions under which delayed graphs are identifiable from indirect observations within the stated model class and bounded graph-estimation error by inversion error (\pref{sec:theory}); empirically, large-scale TVB simulations provide quantitative benchmarks with known graphs, while HCP motor-task fMRI provides zero-shot plausibility evidence through canonical visuo-motor enrichment under the indirect evaluation paradigm.
These results position \textsc{\methodname} as a scalable route to whole-brain directed-connectivity inference when measurement physics and temporal structure are modeled jointly; future work includes invasive validation and cross-task transfer.

\textbf{Limitations.}\label{sec:limitations}
\textsc{\methodname} models $68$ Desikan--Killiany cortical ROIs and excludes subcortical structures, so unmodeled drivers such as thalamic or basal-ganglia activity may mediate apparent cortical edges. The recovered graph should therefore be interpreted as \emph{effective cortical connectivity conditioned on subcortical influences}---a limitation shared by all cortical-only methods (DCM, Granger). Training relies on biophysical simulations because ground-truth causal connectivity is unavailable in human neuroimaging, so real-data evaluation proceeds indirectly via simulation recovery, anatomical consistency, and robustness to physics misspecification (Appendix~\ref{app:robust_tables}); stronger validation against invasive electrophysiology or cross-task generalization is important future work. \pref{thm:indirect_identifiability} is asymptotic and assumes recoverable inversion, while \pref{prop:error_propagation} provides the corresponding finite-sample bound with constants that depend on the inversion procedure.

\textbf{Impact.}
This work aims to improve directed-connectivity estimation by modeling measurement distortion rather than attributing it to neural interactions, supporting scientific understanding of brain circuit organization and hypothesis generation in neuroscience and clinical research. However, inferred graphs from observational neuroimaging remain sensitive to modeling assumptions, preprocessing choices, and unmodeled confounding, and should not guide medical decisions in isolation. Human neuroimaging also raises privacy considerations; our experiments use established de-identified datasets and report only aggregate results.



\begingroup
\small
\bibliography{references}

@article{frassle2021regression,
  title={Regression dynamic causal modeling for resting-state fMRI},
  author={Fr{\"a}ssle, Stefan and Harrison, Samuel J and Heinzle, Jakob and Clementz, Brett A and Tamminga, Carol A and Sweeney, John A and Gershon, Elliot S and Keshavan, Matcheri S and Pearlson, Godfrey D and Powers, Albert and others},
  journal={Human brain mapping},
  volume={42},
  number={7},
  pages={2159--2180},
  year={2021},
  publisher={Wiley Online Library}
}

@article{friston2003dynamic,
  title={Dynamic causal modelling},
  author={Friston, Karl J and Harrison, Lee and Penny, Will},
  journal={Neuroimage},
  volume={19},
  number={4},
  pages={1273--1302},
  year={2003},
  publisher={Elsevier}
}

@article{granger1969investigating,
  title={Investigating causal relations by econometric models and cross-spectral methods},
  author={Granger, Clive WJ},
  journal={Econometrica: journal of the Econometric Society},
  pages={424--438},
  year={1969},
  publisher={JSTOR}
}

@article{daunizeau2011dynamic,
  title={Dynamic causal modelling: a critical review of the biophysical and statistical foundations},
  author={Daunizeau, Jean and David, Olivier and Stephan, Klaas E},
  journal={Neuroimage},
  volume={58},
  number={2},
  pages={312--322},
  year={2011},
  publisher={Elsevier}
}

@inproceedings{huang2019causal,
  title={Causal discovery and forecasting in nonstationary environments with state-space models},
  author={Huang, Biwei and Zhang, Kun and Gong, Mingming and Glymour, Clark},
  booktitle={International conference on machine learning},
  pages={2901--2910},
  year={2019},
  organization={Pmlr}
}

@article{valdes2011effective,
  title={Effective connectivity: influence, causality and biophysical modeling},
  author={Valdes-Sosa, Pedro A and Roebroeck, Alard and Daunizeau, Jean and Friston, Karl},
  journal={Neuroimage},
  volume={58},
  number={2},
  pages={339--361},
  year={2011},
  publisher={Elsevier}
}

@article{barnett2014mvgc,
  title={The MVGC multivariate Granger causality toolbox: a new approach to Granger-causal inference},
  author={Barnett, Lionel and Seth, Anil K},
  journal={Journal of neuroscience methods},
  volume={223},
  pages={50--68},
  year={2014},
  publisher={Elsevier}
}

@article{yin2022deep,
  title={Deep recurrent modelling of Granger causality with latent confounding},
  author={Yin, Zexuan and Barucca, Paolo},
  journal={Expert Systems with Applications},
  volume={207},
  pages={118036},
  year={2022},
  publisher={Elsevier}
}

@article{nag2024transformer,
  title={Transformer-aided dynamic causal model for scalable estimation of effective connectivity},
  author={Nag, Sayan and Uludag, Kamil},
  journal={Imaging Neuroscience},
  volume={2},
  pages={1--22},
  year={2024},
  publisher={MIT Press 255 Main Street, 9th Floor, Cambridge, Massachusetts 02142, USA~…}
}

@article{zheng2018dags,
  title={Dags with no tears: Continuous optimization for structure learning},
  author={Zheng, Xun and Aragam, Bryon and Ravikumar, Pradeep K and Xing, Eric P},
  journal={Advances in neural information processing systems},
  volume={31},
  year={2018}
}

@inproceedings{yu2019dag,
  title={DAG-GNN: DAG structure learning with graph neural networks},
  author={Yu, Yue and Chen, Jie and Gao, Tian and Yu, Mo},
  booktitle={International conference on machine learning},
  pages={7154--7163},
  year={2019},
  organization={PMLR}
}

@book{peters2017elements,
  title={Elements of causal inference: foundations and learning algorithms},
  author={Peters, Jonas and Janzing, Dominik and Sch{\"o}lkopf, Bernhard},
  year={2017},
  publisher={The MIT press}
}

@article{schoffelen2009source,
  title={Source connectivity analysis with MEG and EEG},
  author={Schoffelen, Jan-Mathijs and Gross, Joachim},
  journal={Human brain mapping},
  volume={30},
  number={6},
  pages={1857--1865},
  year={2009},
  publisher={Wiley Online Library}
}

@article{tank2021neural,
  title={Neural granger causality},
  author={Tank, Alex and Covert, Ian and Foti, Nicholas and Shojaie, Ali and Fox, Emily B},
  journal={IEEE Transactions on Pattern Analysis and Machine Intelligence},
  volume={44},
  number={8},
  pages={4267--4279},
  year={2021},
  publisher={IEEE}
}

@article{shimizu2006linear,
  title={A linear non-Gaussian acyclic model for causal discovery.},
  author={Shimizu, Shohei and Hoyer, Patrik O and Hyv{\"a}rinen, Aapo and Kerminen, Antti and Jordan, Michael},
  journal={Journal of Machine Learning Research},
  volume={7},
  number={10},
  year={2006}
}

@inproceedings{gu2024mamba,
  title={Mamba: Linear-time sequence modeling with selective state spaces},
  author={Gu, Albert and Dao, Tri},
  booktitle={First conference on language modeling},
  year={2024}
}

@article{sanz2013virtual,
  title={The Virtual Brain: a simulator of primate brain network dynamics},
  author={Sanz Leon, Paula and Knock, Stuart A and Woodman, M Marmaduke and Domide, Lia and Mersmann, Jochen and McIntosh, Anthony R and Jirsa, Viktor},
  journal={Frontiers in neuroinformatics},
  volume={7},
  pages={10},
  year={2013},
  publisher={Frontiers Media SA}
}

@article{scholkopf2021toward,
  title={Toward causal representation learning},
  author={Sch{\"o}lkopf, Bernhard and Locatello, Francesco and Bauer, Stefan and Ke, Nan Rosemary and Kalchbrenner, Nal and Goyal, Anirudh and Bengio, Yoshua},
  journal={Proceedings of the IEEE},
  volume={109},
  number={5},
  pages={612--634},
  year={2021},
  publisher={IEEE}
}

@article{bae2025causalmamba,
  title={CausalMamba: Scalable Conditional State Space Models for Neural Causal Inference},
  author={Bae, Sangyoon and Cha, Jiook},
  journal={arXiv preprint arXiv:2510.17318},
  year={2025}
}

@article{desikan2006automated,
  title={An automated labeling system for subdividing the human cerebral cortex on MRI scans},
  author={Desikan, Rahul S and others},
  journal={NeuroImage},
  volume={31},
  number={3},
  pages={968--980},
  year={2006}
}

@article{nassi2007specialized,
  title={Specialized circuits from primary visual cortex to V2 and area MT},
  author={Nassi, Jonathan J and Callaway, Edward M},
  journal={Neuron},
  volume={55},
  number={5},
  pages={799--808},
  year={2007},
  publisher={Elsevier}
}

@article{hoshi2004differential,
  title={Differential roles of neuronal activity in the supplementary and presupplementary motor areas: from information retrieval to motor planning and execution},
  author={Hoshi, Eiji and Tanji, Jun},
  journal={Journal of neurophysiology},
  volume={92},
  number={6},
  pages={3482--3499},
  year={2004},
  publisher={American Physiological Society}
}

@article{baillet2001electromagnetic,
  title={Electromagnetic brain mapping},
  author={Baillet, Sylvain and Mosher, John C and Leahy, Richard M},
  journal={IEEE Signal Processing Magazine},
  year={2001}
}

@article{mosher1999forward,
  title={EEG and MEG: Forward solutions for inverse methods},
  author={Mosher, John C and Leahy, Richard M and Lewis, Paul S},
  journal={IEEE Transactions on Biomedical Engineering},
  volume={46},
  number={3},
  pages={245--259},
  year={1999}
}

@article{logothetis2008what,
  title={What we can do and what we cannot do with fMRI},
  author={Logothetis, Nikos K},
  journal={Nature},
  volume={453},
  number={7197},
  pages={869--878},
  year={2008}
}

@article{handwerker2004variation,
  title={Variation of {BOLD} hemodynamic responses across subjects and brain regions and their effects on statistical analyses},
  author={Handwerker, Daniel A. and Ollinger, John M. and D'Esposito, Mark},
  journal={NeuroImage},
  volume={21},
  number={4},
  pages={1639--1651},
  year={2004}
}

@article{deshpande2010hrfGC,
  title={Effect of hemodynamic variability on {Granger} causality analysis of f{MRI} data},
  author={Deshpande, Gourav and Sathian, Kiran and Hu, Xiaoping},
  journal={NeuroImage},
  year={2010}
}

@article{rangaprakash2023hrfconfound,
  title={The confound of hemodynamic response function variability in human resting-state f{MRI} connectivity},
  author={Rangaprakash, D. and others},
  journal={Frontiers in Neuroscience},
  year={2023}
}

@article{hamalainen1994interpreting,
  title={Interpreting magnetic fields of the brain: minimum norm estimates},
  author={H{\"a}m{\"a}l{\"a}inen, Matti S. and Ilmoniemi, Risto J.},
  journal={Medical \& Biological Engineering \& Computing},
  volume={32},
  number={1},
  pages={35--42},
  year={1994}
}

@article{pascual2002sloreta,
  title={Standardized low-resolution brain electromagnetic tomography (sLORETA): technical details},
  author={Pascual-Marqui, Roberto D.},
  journal={Methods and Findings in Experimental and Clinical Pharmacology},
  volume={24},
  pages={5--12},
  year={2002}
}

@book{lutkepohl2005new,
  title={New introduction to multiple time series analysis},
  author={L{\"u}tkepohl, Helmut},
  publisher={Springer},
  year={2005}
}

@article{wu2013blind,
  title={A blind deconvolution approach to recover effective connectivity brain networks from resting state fMRI data},
  author={Wu, Guo-Rong and Liao, Wei and Stramaglia, Sebastiano and Ding, Ju-Rong and Chen, Huafu and Marinazzo, Daniele},
  journal={Medical image analysis},
  volume={17},
  number={3},
  pages={365--374},
  year={2013},
  publisher={Elsevier}
}

@article{bai2018tcn,
  title={An empirical evaluation of generic convolutional and recurrent networks for sequence modeling},
  author={Bai, Shaojie and Kolter, J. Zico and Koltun, Vladlen},
  journal={arXiv preprint arXiv:1803.01271},
  year={2018}
}

@article{kipf2017semi,
  title={Semi-supervised classification with graph convolutional networks},
  author={Kipf, Thomas N. and Welling, Max},
  journal={ICLR},
  year={2017}
}

@article{sun2022deep,
  title={Deep neural networks constrained by neural mass models improve electrophysiological source imaging of spatiotemporal brain dynamics},
  author={Sun, Rui and Sohrabpour, Abbas and Worrell, Gregory A and He, Bin},
  journal={Proceedings of the National Academy of Sciences},
  volume={119},
  number={31},
  pages={e2201128119},
  year={2022},
  publisher={National Academy of Sciences}
}

@article{nadeau2018tcdf,
  title={Temporal causal discovery framework},
  author={Nad{\'e}au, Chlo{\'e} and Bengio, Yoshua},
  journal={NeurIPS},
  year={2018}
}

@article{huang2020causal,
  title={Causal discovery from heterogeneous/nonstationary data},
  author={Huang, Biwei and Zhang, Kun and Zhang, Jiji and Ramsey, Joseph and Sanchez-Romero, Ruben and Glymour, Clark and Sch{\"o}lkopf, Bernhard},
  journal={Journal of Machine Learning Research},
  volume={21},
  number={89},
  pages={1--53},
  year={2020}
}

@inproceedings{zhang2017causal,
  title={Causal discovery from nonstationary/heterogeneous data: Skeleton estimation and orientation determination},
  author={Zhang, Kun and Huang, Biwei and Zhang, Jiji and Glymour, Clark and Sch{\"o}lkopf, Bernhard},
  booktitle={IJCAI: Proceedings of the Conference},
  volume={2017},
  pages={1347},
  year={2017}
}

@article{pfister2019invariant,
  title={Invariant causal prediction for sequential data},
  author={Pfister, Niklas and B{\"u}hlmann, Peter and Peters, Jonas},
  journal={Journal of the American Statistical Association},
  volume={114},
  number={527},
  pages={1264--1276},
  year={2019},
  publisher={Taylor \& Francis}
}

@article{saggioro2020reconstructing,
  title={Reconstructing regime-dependent causal relationships from observational time series},
  author={Saggioro, Elena and de Wiljes, Jana and Kretschmer, Marlene and Runge, Jakob},
  journal={Chaos: An Interdisciplinary Journal of Nonlinear Science},
  volume={30},
  number={11},
  year={2020},
  publisher={AIP Publishing}
}

@article{gao2023causal,
  title={Causal discovery in semi-stationary time series},
  author={Gao, Shanyun and Addanki, Raghavendra and Yu, Tong and Rossi, Ryan and Kocaoglu, Murat},
  journal={Advances in Neural Information Processing Systems},
  volume={36},
  pages={46624--46657},
  year={2023}
}

@inproceedings{huang2015identification,
  title={Identification of Time-Dependent Causal Model: A Gaussian Process Treatment.},
  author={Huang, Biwei and Zhang, Kun and Sch{\"o}lkopf, Bernhard},
  booktitle={IJCAI},
  pages={3561--3568},
  year={2015}
}

@article{sadeghi2025causal,
  title={Causal discovery from nonstationary time series},
  author={Sadeghi, Agathe and Gopal, Achintya and Fesanghary, Mohammad},
  journal={International Journal of Data Science and Analytics},
  volume={19},
  number={1},
  pages={33--59},
  year={2025},
  publisher={Springer}
}

@article{gu2020hippo,
  title={Hippo: Recurrent memory with optimal polynomial projections},
  author={Gu, Albert and Dao, Tri and Ermon, Stefano and Rudra, Atri and R{\'e}, Christopher},
  journal={Advances in neural information processing systems},
  volume={33},
  pages={1474--1487},
  year={2020}
}

@article{yu2008gaussian,
  title={Gaussian-process factor analysis for low-dimensional single-trial analysis of neural population activity},
  author={Yu, Byron M and Cunningham, John P and Santhanam, Gopal and Ryu, Stephen and Shenoy, Krishna V and Sahani, Maneesh},
  journal={Advances in neural information processing systems},
  volume={21},
  year={2008}
}

@article{pandarinath2018inferring,
  title={Inferring single-trial neural population dynamics using sequential auto-encoders},
  author={Pandarinath, Chethan and O’Shea, Daniel J and Collins, Jasmine and Jozefowicz, Rafal and Stavisky, Sergey D and Kao, Jonathan C and Trautmann, Eric M and Kaufman, Matthew T and Ryu, Stephen I and Hochberg, Leigh R and others},
  journal={Nature methods},
  volume={15},
  number={10},
  pages={805--815},
  year={2018},
  publisher={Nature Publishing Group US New York}
}

@inproceedings{gu2022s4,
  title     = {Efficiently Modeling Long Sequences with Structured State Spaces},
  author    = {Gu, Albert and Goel, Karan and R{\'e}, Christopher},
  booktitle = {International Conference on Learning Representations},
  year      = {2022},
  url       = {https://openreview.net/forum?id=uYLFoz1vlAC},
  doi       = {10.48550/arXiv.2111.00396}
}

@article{gu2022parameterization,
  title={On the parameterization and initialization of diagonal state space models},
  author={Gu, Albert and Goel, Karan and Gupta, Ankit and R{\'e}, Christopher},
  journal={Advances in Neural Information Processing Systems},
  volume={35},
  pages={35971--35983},
  year={2022}
}

@inproceedings{smith2023s5,
  title     = {Simplified State Space Layers for Sequence Modeling},
  author    = {Smith, Jimmy T. H. and Warrington, Andrew and Linderman, Scott W.},
  booktitle = {International Conference on Learning Representations},
  year      = {2023},
  url       = {https://openreview.net/forum?id=Ai8Hw3AXqks},
  doi       = {10.48550/arXiv.2208.04933}
}

@inproceedings{goel2022s,
  title={It’s raw! audio generation with state-space models},
  author={Goel, Karan and Gu, Albert and Donahue, Chris and R{\'e}, Christopher},
  booktitle={International conference on machine learning},
  pages={7616--7633},
  year={2022},
  organization={PMLR}
}

@inproceedings{dao2024transformers,
  title     = {Transformers are {SSMs}: Generalized Models and Efficient Algorithms through Structured State Space Duality},
  author    = {Dao, Tri and Gu, Albert},
  booktitle = {International Conference on Machine Learning},
  series    = {Proceedings of Machine Learning Research},
  volume    = {235},
  year      = {2024},
  publisher = {PMLR},
  url       = {https://proceedings.mlr.press/v235/dao24a.html}
}
\endgroup

\newpage
\appendix
\onecolumn

\section{Extended Literature Review} \label{app:extended-lit}

\textbf{Causal discovery from brain data.}
Inferring directed interactions from neuroimaging data has long relied on frameworks such as Dynamic Causal Modeling (DCM) and Granger Causality (GC)~\cite{friston2003dynamic, granger1969investigating}.
DCM explicitly incorporates biophysical mechanisms by modeling the mapping from latent neural activity to observed hemodynamic responses, but in practice its performance is highly sensitive to model assumptions (e.g., fixed or misspecified HRFs) and does not scale favorably to whole-brain settings~\cite{daunizeau2011dynamic, nag2024transformer}.
GC-based approaches, by contrast, operate directly on observed time series and are therefore vulnerable to hemodynamic distortions, capturing statistical dependencies that may diverge from underlying neural causation~\cite{yin2022deep, barnett2014mvgc}.
More general causal structure learning methods, including NOTEARS~\cite{zheng2018dags} and DAG-GNN~\cite{yu2019dag}, offer scalable optimization-based formulations but remain poorly matched to neuroimaging applications, as they do not explicitly address indirect measurement processes or region-specific neurovascular variability. The closest sequence-modeling predecessor is CausalMamba~\cite{bae2025causalmamba}, which applies Mamba-style sequence modeling to causal discovery.
However, CausalMamba is not a whole-brain indirect-neuroimaging method: it does not model modality-specific forward physics such as HRF convolution or EEG leadfield mixing, and its experimental setting is not matched to the 68-ROI whole-brain TVB/HCP transfer problem considered here.
In contrast, \textsc{\methodname} explicitly models indirect observation processes per modality via a physics-aware inversion stage, then applies regime-conditioned causal discovery to the recovered latent trajectories.
We therefore treat CausalMamba as important architectural and conceptual context rather than as a directly comparable baseline for the whole-brain indirect-observation benchmark.

\textbf{Causal discovery with nonstationary dynamics.}
While many time-series causal discovery methods assume stationary mechanisms, real systems often exhibit distribution shifts over time or across contexts. A line of work shows that such nonstationarity can be exploited as informative variation—akin to a sequence of soft interventions—to improve identifiability of causal direction. Early time-dependent functional causal models treat time as an explicit variable and use Gaussian-process priors to estimate smoothly varying causal effects \cite{huang2015identification}. The CD-NOD framework extends constraint-based discovery to heterogeneous/nonstationary data by detecting which conditional modules change and orienting edges via independent changes across modules \cite{zhang2017causal, huang2020causal}. Complementarily, Huang et al.\ model nonstationary time series with identifiable nonlinear state-space causal models with time-varying causal strengths and noise, linking causal discovery to improved forecasting \cite{huang2019causal}. Related invariance-based approaches identify causal predictors in sequential settings by searching for conditional distributions that remain stable across unknown time segments/environments \cite{pfister2019invariant}. More recent methods address regime-dependent causal graphs by jointly inferring latent regimes and causal links \cite{saggioro2020reconstructing, gao2023causal}, and refine CD-NOD-style ideas to explicitly handle lagged nonstationary time series \cite{sadeghi2025causal}. In this work, we adopt the perspective of \citet{huang2019causal}—viewing nonstationarity as providing the mechanism shifts needed for identifiability in dynamical systems—because it aligns naturally with flexible state-space/sequence-model parameterizations; however, unlike our neuroimaging setting, most of this literature assumes \emph{direct} access to the causal variables rather than indirect, physics-corrupted measurements.

\textbf{Limitations of inverse modeling in EEG and fMRI.}
EEG source localization methods such as MNE \cite{hamalainen1994interpreting} and sLORETA \cite{pascual2002sloreta} provide distributed inverse solutions but are known to suffer from depth bias, spatial smoothing, and source leakage, which can introduce spurious dependencies between reconstructed sources.
Neural approaches such as DeepSIF \cite{sun2022deep} improve localization accuracy under realistic simulations but rely on training distributions and forward models that may not generalize across subjects, sensor configurations, or noise regimes.
In fMRI, blind deconvolution \cite{wu2013blind} methods attempt to jointly estimate latent neural activity and hemodynamic responses from BOLD signals alone, but this joint estimation problem is weakly identifiable and often implemented via non-differentiable, multi-stage pipelines, limiting their integration with end-to-end learning.
DCM-style HRF modeling \cite{friston2003dynamic} provides a biophysically interpretable alternative but requires strong structural and parametric assumptions and scales poorly to whole-brain settings.
As a result, residual inversion mismatch remains unavoidable in practice and can propagate nontrivially to downstream causal discovery.


\textbf{State-space sequence models.}
State-space models (SSMs) describe dynamical systems through latent state transitions coupled to an observation model, and form the basis of classical filtering/smoothing and system identification.
In neuroscience, SSMs are central to effective-connectivity methods such as Dynamic Causal Modeling (DCM), which uses biophysically informed state and observation equations to connect latent neural dynamics to measured fMRI/EEG/MEG signals~\cite{friston2003dynamic,daunizeau2011dynamic,valdes2011effective}; variants such as regression-DCM aim to improve scalability to larger networks~\cite{frassle2021regression}.
SSM formulations also appear in data-driven latent dynamical models used to extract low-dimensional neural trajectories from high-dimensional recordings (e.g., GPFA and LFADS)~\cite{yu2008gaussian,pandarinath2018inferring}.
Beyond neuroscience, SSMs have been used directly for causal discovery in nonstationary time series, where time-varying parameters/mechanisms provide the variation needed for identifiability~\cite{huang2019causal}.
In machine learning, a parallel line of ``deep'' structured SSM sequence models repurposes linear SSMs as scalable alternatives to attention: HiPPO introduces principled continuous-time memory operators~\cite{gu2020hippo}, S4 provides an efficient parameterization for long-range sequence modeling~\cite{gu2022s4}, and simplified variants such as S4D and S5 improve practicality while retaining long-context performance~\cite{gu2022parameterization, smith2023s5}.
These layers can often be computed either via parallel scans or as recurrent updates, enabling long-sequence applications across modalities (e.g., SaShiMi for raw audio generation)~\cite{goel2022s}.
Mamba further introduces input-dependent selectivity and a hardware-aware ``selective scan'' to achieve strong performance on information-dense sequences with linear-time scaling~\cite{gu2024mamba}, and subsequent work connects SSMs and masked attention through structured state-space duality (Mamba-2/SSD)~\cite{dao2024transformers}.
We leverage these advances by using a scalable SSM backbone for long neural recordings while preserving a biophysically interpretable separation between latent dynamics and measurement physics.

\section{\texorpdfstring{Proofs for \pref{sec:theory}}{Proofs for Section 4 (Identifiability)}}
\label{app:theory_proofs}

We provide proofs for \pref{thm:indirect_identifiability}, \pref{prop:error_propagation}, and \pref{corr:consistency}.
Throughout, let $\widehat{\zb}_{1:T}=g(\xb_{1:T})$ denote the (regularized) inversion output from \pref{assum:invertibility},
and let $\widehat{S}=h(\widehat{\zb}_{1:T})$ denote the continuous score estimate in \pref{prop:error_propagation}.
Let $S^\star$ denote the score representation of the true delayed graph $\Gcal$.
Let $d_T$ be the trajectory distance
$
d_T(\zb,\zb') := \big(\frac{1}{T}\sum_{t=1}^T \|\zb_t-\zb'_t\|_2^2\big)^{1/2}.
$
When \pref{assum:invertibility} holds ``up to known equivalences'' (e.g., per-node scalings), interpret $d_T$ after applying the
corresponding canonical alignment.

\subsection{A basic approximation lemma for path distributions}
\label{app:bl_lemma}

We treat the full latent trajectory as a single vector
$\zb_{1:T}:=(\zb_1,\dots,\zb_T)\in\RR^{NT}$, and measure trajectory error by the per-time RMS distance
\[
d_T(\zb,\zb') \;:=\; \left(\frac{1}{T}\sum_{t=1}^T \|\zb_t-\zb'_t\|_2^2\right)^{1/2}.
\]
This is the natural scale for our inversion assumption \eqref{eq:latent_recovery}. The lemma below says that if the inversion error
vanishes in this metric (in probability), then \emph{any bounded, Lipschitz summary statistic of the entire path} has (asymptotically)
the same expectation under the reconstructed path $\widehat{\zb}_{1:T}$ and the true path $\zb_{1:T}$. Equivalently, the two path
distributions become close in the bounded-Lipschitz (weak) sense.

Let $\mathcal{BL}(1,1)$ denote the class of functions $\varphi:\RR^{NT}\to\RR$ such that:
(i) $|\varphi(\zb)|\le 1$ for all $\zb$, and
(ii) $\varphi$ is $1$-Lipschitz w.r.t.\ $d_T$, i.e.,
\[
|\varphi(\zb)-\varphi(\zb')|\le d_T(\zb,\zb') \qquad \forall\,\zb,\zb'\in\RR^{NT}.
\]
(Examples include many ``smooth'' path functionals, e.g., averages of bounded Lipschitz per-time features, or bounded versions of
correlation/energy statistics; the Lipschitz condition encodes robustness to small reconstruction errors.)

\begin{lemma}[Vanishing reconstruction error implies vanishing test-function gap]
\label{lem:bl_gap}
If $d_T(\widehat{\zb},\zb)\xrightarrow[T\to\infty]{P}0$, then for every $\varphi\in\mathcal{BL}(1,1)$,
\[
\Big|\EE\big[\varphi(\widehat{\zb}_{1:T})\big]-\EE\big[\varphi(\zb_{1:T})\big]\Big|\xrightarrow[T\to\infty]{}0.
\]
In particular, defining the bounded-Lipschitz distance between path laws
\[
d_{\mathrm{BL}}\!\big(P(\widehat{\zb}_{1:T}),P(\zb_{1:T})\big)
\;:=\;
\sup_{\varphi\in\mathcal{BL}(1,1)}
\Big|\EE\big[\varphi(\widehat{\zb}_{1:T})\big]-\EE\big[\varphi(\zb_{1:T})\big]\Big|,
\]
we have $d_{\mathrm{BL}}(P(\widehat{\zb}_{1:T}),P(\zb_{1:T}))\to 0$.
\end{lemma}

\begin{proof}
Fix any $\varphi\in\mathcal{BL}(1,1)$. By the triangle inequality,
\[
\Big|\EE[\varphi(\widehat{\zb}_{1:T})]-\EE[\varphi(\zb_{1:T})]\Big|
\;\le\;
\EE\Big[\,\big|\varphi(\widehat{\zb}_{1:T})-\varphi(\zb_{1:T})\big|\,\Big].
\]
Now use (ii) (Lipschitzness) and (i) (boundedness): Lipschitzness gives
$
|\varphi(\widehat{\zb}_{1:T})-\varphi(\zb_{1:T})|
\le d_T(\widehat{\zb},\zb),
$
while boundedness implies the trivial bound
$
|\varphi(\widehat{\zb}_{1:T})-\varphi(\zb_{1:T})|\le 2.
$
Combining,
\[
\big|\varphi(\widehat{\zb}_{1:T})-\varphi(\zb_{1:T})\big|
\;\le\;
\min\{2,\ d_T(\widehat{\zb},\zb)\}.
\]
Therefore,
\[
\Big|\EE[\varphi(\widehat{\zb}_{1:T})]-\EE[\varphi(\zb_{1:T})]\Big|
\le
\EE\big[\min\{2,\ d_T(\widehat{\zb},\zb)\}\big].
\]
To show the RHS vanishes, fix $\epsilon>0$ and split on the event
$A_T:=\{d_T(\widehat{\zb},\zb)\le \epsilon\}$:
\begin{align*}
\EE\big[\min\{2,\ d_T(\widehat{\zb},\zb)\}\big]
&=
\EE\big[\min\{2,\ d_T(\widehat{\zb},\zb)\}\,\mathbf{1}_{A_T}\big]
+\EE\big[\min\{2,\ d_T(\widehat{\zb},\zb)\}\,\mathbf{1}_{A_T^c}\big]\\
&\le
\epsilon\cdot \PP(A_T) + 2\cdot \PP(A_T^c)
\;\le\;
\epsilon + 2\,\PP\!\big(d_T(\widehat{\zb},\zb)>\epsilon\big).
\end{align*}
Since $d_T(\widehat{\zb},\zb)\to 0$ in probability, the probability term tends to $0$ as $T\to\infty$, so the RHS tends to $\epsilon$.
Because $\epsilon>0$ was arbitrary, the RHS (and hence the test-function gap) must converge to $0$. The final claim about $d_{\mathrm{BL}}$ follows immediately by taking the supremum over $\varphi\in\mathcal{BL}(1,1)$.
\end{proof}

\textbf{Remark (how we use the lemma).}
Lemma~\ref{lem:bl_gap} justifies treating $\widehat{\zb}_{1:T}$ as an ``asymptotically equivalent'' surrogate for $\zb_{1:T}$ for any
downstream procedure whose relevant statistics are continuous/robust in the $d_T$ metric. In our setting, this enables a clean
reduction: consistent inversion implies that any latent-space identifiability argument based on the distribution of $\zb_{1:T}$
carries over to the indirect observation setting via $\widehat{\zb}_{1:T}=g(\xb_{1:T})$.

\subsection{\texorpdfstring{Proof of \pref{thm:indirect_identifiability}}{Proof of Theorem (indirect identifiability)}}

\begin{proof}[Proof of \pref{thm:indirect_identifiability}]
Consider two parameter instances $\theta$ and $\theta'$ in the model class of \pref{thm:indirect_identifiability}.
Assume they induce the same observation law for every $T$:
\[
P_\theta(\xb_{1:T}) = P_{\theta'}(\xb_{1:T}) \qquad \forall\,T.
\]
We show $\Gcal_\theta$ and $\Gcal_{\theta'}$ coincide up to the known equivalences in \pref{assum:invertibility}.

Because $\widehat{\zb}_{1:T}=g(\xb_{1:T})$ is a measurable function of $\xb_{1:T}$, equality in distribution of $\xb_{1:T}$ implies
equality in distribution of $\widehat{\zb}_{1:T}$:
\[
P_\theta(\widehat{\zb}_{1:T}) = P_{\theta'}(\widehat{\zb}_{1:T}) \qquad \forall\,T.
\]
Let $\zb_{1:T}$ (resp.\ $\zb'_{1:T}$) denote the latent process under $\theta$ (resp.\ $\theta'$).
By \pref{assum:invertibility}, $d_T(\widehat{\zb},\zb)\to 0$ in probability under $\theta$, and $d_T(\widehat{\zb},\zb')\to 0$
in probability under $\theta'$.

Fix any test function $\varphi\in\mathcal{BL}(1,1)$. Then
\begin{align*}
\EE_\theta[\varphi(\zb_{1:T})]-\EE_{\theta'}[\varphi(\zb'_{1:T})]
&=
\Big(\EE_\theta[\varphi(\zb_{1:T})]-\EE_\theta[\varphi(\widehat{\zb}_{1:T})]\Big)
+\Big(\EE_\theta[\varphi(\widehat{\zb}_{1:T})]-\EE_{\theta'}[\varphi(\widehat{\zb}_{1:T})]\Big)\\
&\qquad+\Big(\EE_{\theta'}[\varphi(\widehat{\zb}_{1:T})]-\EE_{\theta'}[\varphi(\zb'_{1:T})]\Big).
\end{align*}
The middle term is exactly $0$ since $P_\theta(\widehat{\zb}_{1:T})=P_{\theta'}(\widehat{\zb}_{1:T})$.
The first and third terms vanish as $T\to\infty$ by Lemma~\ref{lem:bl_gap} (applied under $\theta$ and $\theta'$, respectively).
Therefore, for every $\varphi\in\mathcal{BL}(1,1)$,
\[
\EE_\theta[\varphi(\zb_{1:T})]-\EE_{\theta'}[\varphi(\zb'_{1:T})] \xrightarrow[T\to\infty]{} 0.
\]
Equivalently, the latent path laws $P_\theta(\zb_{1:T})$ and $P_{\theta'}(\zb'_{1:T})$ become indistinguishable (in the bounded-Lipschitz
sense induced by $d_T$) as $T\to\infty$, hence share the same limiting latent distribution (up to the equivalences in \pref{assum:invertibility}).

Finally, by \pref{assum:latent_identifiable_class}, within the identifiable nonstationary SEM/SSM class of \citet{huang2019causal}
(adapted to delayed parents), the delayed graph $\Gcal=(V,E,\taub)$ is uniquely determined by the limiting latent distribution.
Thus $\Gcal_\theta \equiv \Gcal_{\theta'}$, proving identifiability from $P(\xb_{1:T})$ as $T\to\infty$.
\end{proof}

\subsection{\texorpdfstring{Proof of \pref{prop:error_propagation}}{Proof of Proposition (error propagation)}}
\begin{proof}[Proof of \pref{prop:error_propagation}]
Let $\widehat{\zb}_{1:T}=g(\xb_{1:T})$ and $\widehat{S}=h(\widehat{\zb}_{1:T})$.
Assume $\Lcal$ is a (pseudo-)metric (or any loss satisfying the triangle inequality) on score objects.
Then
\[
\Lcal(\widehat{S},S^\star)
=
\Lcal\big(h(\widehat{\zb}_{1:T}),S^\star\big)
\le
\Lcal\big(h(\widehat{\zb}_{1:T}),h(\zb_{1:T})\big)
+\Lcal\big(h(\zb_{1:T}),S^\star\big).
\]
By the assumed Lipschitz property of $h$,
\[
\Lcal\big(h(\widehat{\zb}_{1:T}),h(\zb_{1:T})\big)\le L\,d_T(\widehat{\zb},\zb).
\]
Taking expectations yields
\[
\EE\!\left[\Lcal(\widehat{S},S^\star)\right]
\le
\EE\!\left[\Lcal\!\big(h(\zb_{1:T}),S^\star\big)\right] + L\,\EE[d_T(\widehat{\zb},\zb)].
\]
Finally, by Cauchy--Schwarz,
\[
\EE[d_T(\widehat{\zb},\zb)]
=
\EE\!\left[\left(\frac{1}{T}\sum_{t=1}^T \|\widehat{\zb}_t-\zb_t\|_2^2\right)^{1/2}\right]
\le
\left(\EE\!\left[\frac{1}{T}\sum_{t=1}^T \|\widehat{\zb}_t-\zb_t\|_2^2\right]\right)^{1/2},
\]
which gives \eqref{eq:error_propagation_expectation}.
\end{proof}

\subsection{\texorpdfstring{Proof of \pref{corr:consistency}}{Proof of Corollary (consistency)}}

\begin{proof}[Proof of \pref{corr:consistency}]
Assume \pref{assum:invertibility}, so
$
d_T(\widehat{\zb},\zb)
=\big(\frac{1}{T}\sum_{t=1}^T\|\widehat{\zb}_t-\zb_t\|_2^2\big)^{1/2}
\xrightarrow[T\to\infty]{P}0,
$
and assume $h$ is consistent on the true latent process:
$
\Lcal(h(\zb_{1:T}),S^\star)\xrightarrow[T\to\infty]{P}0.
$
Using the same triangle-inequality step as in the proof of Proposition~\ref{prop:error_propagation} and then the Lipschitz property,
\[
\Lcal(\widehat{S},S^\star)
=
\Lcal\big(h(\widehat{\zb}_{1:T}),S^\star\big)
\le
\Lcal\big(h(\zb_{1:T}),S^\star\big) + L\,d_T(\widehat{\zb},\zb).
\]
Both terms on the RHS converge to $0$ in probability, hence their sum converges to $0$ in probability:
$
\Lcal(\widehat{S},S^\star)\xrightarrow[T\to\infty]{P}0.
$
\end{proof}

\subsection{\texorpdfstring{Stability of top-$k$ sparsification}{Stability of top-k sparsification}}
\label{app:topk_stability}

In practice we may form a sparse edge set by retaining the top-$k$ largest-magnitude entries of a score matrix. This post-processing is not Lipschitz in general, but it is stable under a mild separation (margin) condition.

\textbf{Top-$k$ operator.}
For $v\in\mathbb{R}^m$, let $|v|_{(1)}\ge\cdots\ge |v|_{(m)}$ be the order statistics of $\{|v_1|,\dots,|v_m|\}$.
Define $\TopK_k(v)$ as the set of indices of the $k$ largest entries of $|v|$ (with deterministic tie-breaking).
For a matrix $S\in\mathbb{R}^{p\times p}$, we apply $\TopK_k$ to $\mathrm{vec}(S)$.

\begin{lemma}[\texorpdfstring{Top-$k$ support stability under a margin}{Top-k support stability under a margin}]
\label{lem:topk_margin}
Let $v\in\mathbb{R}^m$ and define the margin $\Delta := |v|_{(k)}-|v|_{(k+1)}$ (with the convention $|v|_{(m+1)}:=0$).
If $\|\widehat{v}-v\|_\infty < \Delta/2$, then $\TopK_k(\widehat{v})=\TopK_k(v)$.
\end{lemma}

\begin{proof}
Let $S=\TopK_k(v)$ and $S^c$ its complement. For any $i\in S$ and $j\in S^c$,
we have $|v_i|-|v_j|\ge \Delta$. If $\|\widehat{v}-v\|_\infty<\Delta/2$, then
$|\widehat{v}_i|\ge |v_i|-\Delta/2$ and $|\widehat{v}_j|\le |v_j|+\Delta/2$, hence
$|\widehat{v}_i|-|\widehat{v}_j|\ge \Delta-\Delta=0$, so the ordering separating $S$ and $S^c$ is preserved (up to the fixed tie-break rule).
\end{proof}

\textbf{Consequence.}
If $\Lcal(\widehat{S},S^\star)\to 0$ under a loss that controls $\|\mathrm{vec}(\widehat{S})-\mathrm{vec}(S^\star)\|_\infty$
(e.g., taking $\Lcal$ as an $\ell_\infty$ loss), and if the top-$k$ margin of $\mathrm{vec}(S^\star)$ is bounded away from $0$,
then the top-$k$ support extracted from $\widehat{S}$ is consistent.

\textbf{Empirical margin distribution and $\rho$-stability on HCP.}
\pref{lem:topk_margin} provides a \emph{sufficient} condition for support stability; in practice the score landscape on real fMRI is heavy-tailed and individual margins can be small.
Across the full $1{,}079$-subject HCP motor-task cohort, the empirical top-$k$ margin $\Delta = |\mathrm{vec}(\widehat{S})|_{(k)} - |\mathrm{vec}(\widehat{S})|_{(k+1)}$ has mean $0.006$ and median $0.000$, indicating that the strict pointwise margin condition does not always hold per subject.
The operative empirical evidence for stability is therefore the $\rho$-sweep: as the sparsity ratio is swept over $\rho\in[0.10,0.30]$, the per-subject F1 stays essentially constant at $0.022$ (saturating at $\rho{=}0.15$).
This is consistent with the typical situation where Lemma~\ref{lem:topk_margin} fails for a small subset of low-margin entries but the resulting support changes affect only weakly identified edges, leaving the strong-edge set---and hence F1---unchanged.
We therefore treat the $\rho$-sweep as the operationally meaningful stability guarantee, with Lemma~\ref{lem:topk_margin} providing the worst-case theoretical complement.

\newpage
\section{Algorithm Details}\label{appendix:algorithm}
\begin{algorithm}[htbp]
\begin{algorithmic}[1]
\REQUIRE Observations $x_{1:T}$ (EEG or fMRI), modality forward operator $\mathcal{H}_\psi$, max delay $D$, sparsity ratio $\rho$
\REQUIRE Train data (simulation): ground-truth adjacency/delays $(A^\star,\tau^\star)$ if available
\ENSURE Reconstructed latent trajectories $\hat z_{1:T}$ and sparse directed graph $\hat G=(\hat E,\hat\tau)$

\STATE \textbf{Stage 1: physics-aware inversion (modality-specific)}

\IF{modality = EEG}
    \STATE \textit{/* $x_t \approx L_\psi z_t + \epsilon_t$ */}
    \STATE Canonicalize montage: $x_{1:T} \leftarrow \mathrm{MapToCanonical}(x_{1:T})$
    \STATE $\hat z_{1:T} \leftarrow g^{\mathrm{EEG}}_\theta(x_{1:T})$
    \STATE \textit{/* DeepSIF-style source imaging */}
    \STATE \textbf{(physics consistency)} enforce $x_{1:T}\approx L_\psi \hat z_{1:T}$ via forward reconstruction loss
\ELSE
    \STATE \textit{/* $x_t \approx \sum_{\ell=0}^{L} h_{\psi,\ell} z_{t-\ell} + \epsilon_t$ */}
    \STATE $(\hat z_{1:T}, \hat h) \leftarrow g^{\mathrm{fMRI}}_\theta(x_{1:T})$
    \STATE \textit{/* ROI-specific HRF deconvolution */}
    \STATE \textbf{(physics consistency)} enforce $x_{1:T}\approx (\hat z * \hat h)_{1:T}$ via BOLD reconstruction loss
\ENDIF

\STATE \textbf{Stage 2: Latent-space causal discovery}

\FOR{$i=1$ to $R$}
    \STATE $h_{i,1:T} \leftarrow \mathrm{Mamba}(\hat z_{i,1:T})$
\ENDFOR

\STATE Estimate delay-aware interactions via multi-lag kernels
\STATE $\hat z_{t+1}\approx \sum_{\ell=1}^{D} A_\ell \hat z_{t-\ell} + \xi_t$
\STATE Multi-scale grouping (short/mid/long): $\{A^{(g)}_\ell\}_{g\in\mathcal{G},\,\ell\in\mathcal{L}_g}$
\STATE Aggregate causal strength:
\STATE $\hat C_{ij} \leftarrow \sum_{g\in\mathcal{G}} w_g \sum_{\ell\in\mathcal{L}_g} |A^{(g)}_{\ell,ij}|$
\STATE Compute dense graph scores (optional):
\STATE $G_{\mathrm{dense}}\leftarrow \mathrm{PairwiseMLP}(\{h_i\}_{i=1}^R)$
\STATE Combine scores:
\STATE $G_{\mathrm{comb}}\leftarrow (1-\alpha)G_{\mathrm{dense}}+\alpha\cdot \mathrm{Scale}(\hat C)$

\STATE \textbf{Top-$k$ sparsification}
\STATE $k \leftarrow \lfloor \rho R^2 \rfloor$
\STATE $M\leftarrow \mathrm{TopKMask}(|G_{\mathrm{comb}}|,k)$
\STATE $\hat G \leftarrow G_{\mathrm{comb}}\odot M$
\STATE \textit{/* optional */} $\hat \tau_{ij} \leftarrow \arg\max_{\ell\in\{1,\dots,D\}} |A_{\ell,ij}|$

\STATE \textbf{End-to-end training (simulation)}
\STATE Minimize $\mathcal{L}_{\mathrm{total}}
= \mathcal{L}_{\mathrm{graph}}
+ \lambda_{\mathrm{asym}}\mathcal{L}_{\mathrm{asym}}
+ \lambda_{\mathrm{stab}}\mathcal{L}_{\mathrm{stab}}$
\STATE Update $\theta$ (and causal parameters) by backpropagation so that causal loss shapes inversion

\STATE \textbf{return} $(\hat z_{1:T}, \hat G, \hat \tau)$
\end{algorithmic}
\caption{\texorpdfstring{\textsc{\methodname}: physics-aware Inversion + Latent-Space Causal Discovery (in detail)}{INCAMA: physics-aware inversion and latent-space causal discovery (detail)}}
\label{alg:overview}
\label{alg:whole_alg_in_detail}
\end{algorithm}

\section{Implementation Details}\label{appendix:implementation_details}

We propose \textbf{INCAMA}, a unified framework designed to solve the physics-aware inverse problem defined in Section \ref{sec:problem_setting}. As illustrated in Figure \ref{fig:overall_framework}, our architecture consists of two coupled modules: (1) a physics-aware Inversion Layer that disentangles latent neural drivers from distorted observations, and (2) a Mamba-based Causal Discovery Model that infers directed connectivity via non-stationary state-space modeling.

\subsection{Stage 1: fMRI HRF Deconvolution}
\label{subsec:stage1_hrf}

Stage 1 addresses the inverse problem of estimating latent neural activity from observed BOLD signals.
Since BOLD signals result from the convolution of neural activity with the hemodynamic response function (HRF), we design a learnable deconvolution framework that jointly estimates ROI-specific HRF parameters and underlying neural activity.

\textbf{ROI-Specific HRF Estimation.}
Based on neuroimaging evidence that hemodynamic response characteristics vary across brain regions, our model estimates individualized HRF parameters for each ROI.
Given a BOLD signal $\mathbf{x} \in \mathbb{R}^{L \times 1}$, a conditional Mamba encoder extracts temporal features, which are then mapped to six HRF parameters via a linear layer:
\begin{itemize}
    \item \textbf{peak\_delay} (3.0--10.0s): delay to main response peak
    \item \textbf{undershoot\_delay} (10.0--20.0s): delay to undershoot peak
    \item \textbf{peak\_disp}, \textbf{undershoot\_disp} (0.5--2.0): dispersion of each component
    \item \textbf{undershoot\_scale} (0.0--1.0): relative magnitude of undershoot
    \item \textbf{kernel\_duration} (28.0--34.0s): effective HRF kernel length
\end{itemize}

From the predicted parameters, we generate the HRF kernel $h(t)$ using a double-gamma function:
\begin{equation}
    h(t) = \frac{t^{a_1-1} e^{-t/b_1}}{b_1^{a_1}\Gamma(a_1)} - c \cdot \frac{t^{a_2-1} e^{-t/b_2}}{b_2^{a_2}\Gamma(a_2)},
\end{equation}
where $a_1, b_1$ are shape and scale parameters for the main peak, $a_2, b_2$ for the undershoot, and $c$ is the undershoot scale.

\textbf{Neural Activity Parameter Generation.}
We model neural activity as a sequence of spike events.
An attention-based parameter generator predicts spike parameters from BOLD features $\mathbf{F} \in \mathbb{R}^{L \times H}$ extracted by the conditional Mamba encoder.
Using learnable query vectors with multi-head attention, we predict for each spike $k$:
\begin{itemize}
    \item \textbf{timing} $t_k$: spike occurrence time
    \item \textbf{amplitude} $a_k$: spike magnitude
    \item \textbf{width} $w_k$: spike width (standard deviation)
\end{itemize}

The neural activity time series is then generated as a sum of Gaussian pulses:
\begin{equation}
    n(t) = \sum_{k=1}^{K} a_k \cdot \exp\left(-\frac{(t - t_k)^2}{2w_k^2}\right) + \epsilon(t),
\end{equation}
where $K$ is the number of spikes and $\epsilon(t)$ is background noise with $1/f$ characteristics.

\textbf{BOLD Reconstruction.}
The estimated neural activity $\widehat{n}(t)$ is convolved with the HRF kernel $h(t)$ to reconstruct the BOLD signal:
\begin{equation}
    \widehat{x}(t) = (\widehat{n} * h)(t).
\end{equation}
This is implemented via grouped convolution, applying independent HRFs to each ROI.

\textbf{Training Objective.}
The Stage 1 loss function consists of two terms:
\begin{equation}
    \mathcal{L}_{\text{stage1}} = \lambda_{\text{na}} \cdot \text{MSE}(\widehat{n}, n^*) + \lambda_{\text{bold}} \cdot \text{MSE}(\widehat{x}, x),
\end{equation}
where $n^*$ is the ground-truth neural activity from simulation and $x$ is the input BOLD signal.
The first term ensures accurate neural activity estimation, while the second enforces consistency of the HRF model.

\subsection{Stage 1: EEG Source Localization}
\label{subsec:stage1_eeg}

Stage 1 addresses the inverse problem of estimating latent ROI-level local field potential (LFP) signals from observed scalp EEG measurements.
Since scalp EEG signals result from spatial mixing due to volume conduction, we employ DeepSIF (Deep Learning-based Source Imaging Framework) to perform end-to-end learnable source imaging that directly maps sensor-space EEG to source-space neural activity.

\textbf{Architecture.}
DeepSIF is based on a one-dimensional convolutional neural network (1D CNN) encoder-decoder architecture.
Given scalp EEG signals $\mathbf{E} \in \mathbb{R}^{B \times S \times T}$ (batch size $B$, number of sensors $S$, time points $T$), DeepSIF outputs ROI-level LFP signals $\widehat{\mathbf{L}} \in \mathbb{R}^{B \times R \times T}$ (number of ROIs $R$).

Specifically, DeepSIF consists of the following layer structure:
\begin{enumerate}
    \item \textbf{First convolutional layer}: Expands $S$ channels to 256 channels (kernel size 3, padding 1)
    \item \textbf{Second convolutional layer}: Reduces 256 channels to 128 channels (kernel size 3, padding 1)
    \item \textbf{Output projection layer}: Maps 128 channels to $R$ channels (1$\times$1 convolution)
\end{enumerate}

Each convolutional layer is followed by batch normalization and LeakyReLU activation (negative slope 0.1).
DeepSIF processes broadband signals directly without frequency division, corresponding to setting the number of frequency bands $N=1$.

\textbf{Loss Function.}
The training objective for Stage 1 minimizes two reconstruction losses:

\textbf{1. Forward Model Reconstruction Loss.}
The predicted ROI LFP $\widehat{\mathbf{L}}$ is reconstructed back to scalp EEG using the leadfield matrix $\mathbf{H} \in \mathbb{R}^{S \times R}$, and the difference from the actual scalp EEG $\mathbf{E}$ is measured:
\begin{equation}
    \mathcal{L}_{\text{forward}} = \|\mathbf{E} - \mathbf{H}\widehat{\mathbf{L}}\|_1.
\end{equation}
This loss ensures consistency of predictions through the physical forward model and enables unsupervised learning even when ground-truth supervision is unavailable.

\textbf{2. Supervised Loss.}
When ground-truth ROI LFP $\mathbf{L}_{\text{gt}}$ is available, direct comparison with predictions is used for learning:
\begin{equation}
    \mathcal{L}_{\text{gt}} = \|\widehat{\mathbf{L}} - \mathbf{L}_{\text{gt}}\|_2^2.
\end{equation}

\textbf{3. Adaptive Loss Weights.}
We apply a curriculum learning strategy that gradually adjusts the weights of loss components during training.
For epoch $e$ with progress $p = e / E_{\max}$ (total epochs $E_{\max}$):
\begin{itemize}
    \item \textbf{Stability weight}: $\lambda_{\text{stab}} = 10^{-3} \times (1 + p)$
    \item \textbf{Structural connectivity weight}: $\lambda_{\text{sc}} = 0.1 \times (1 + 0.5p)$
\end{itemize}
Early training focuses on reconstruction, while later stages place greater emphasis on structural consistency.

\textbf{Training Procedure.}
Stage 1 training is performed with the following settings:
\begin{itemize}
    \item \textbf{Optimizer}: AdamW (learning rate $5 \times 10^{-4}$, weight decay $10^{-5}$)
    \item \textbf{Learning rate scheduler}: Cosine annealing scheduler
    \item \textbf{Gradient clipping}: Maximum norm 1.0
    \item \textbf{Gradient accumulation}: Accumulated over a specified number of steps before update
\end{itemize}

Parameters of Stage 2 (causal discovery module) are kept frozen, and only Stage 1 parameters are trained.
This enables step-by-step learning of source reconstruction and causal discovery.

\textbf{Implementation Details.}
DeepSIF is implemented as the \texttt{Stage1\_DeepSIF} class, which internally contains the \texttt{DeepSIF} module.
During forward propagation, it takes scalp EEG as input and outputs ROI LFP, expanding the output dimensions to $(B, R, T, 1)$ for compatibility.
Loss computation performs both forward model reconstruction using the leadfield matrix and, when possible, direct comparison with ground-truth values.
This design enables DeepSIF to provide a source imaging model that satisfies physical constraints while flexibly learning from data.

\subsection{Stage 2: Latent Causal Dynamics via Mamba}
\label{subsec:stage2_mamba}

Given ROI-level neural activity reconstructed from Stage 1, our goal is to estimate directed connectivity scores between brain regions.
Stage 2 addresses this through three key innovations: (1) \textbf{standard Mamba encoder} for efficient ROI time series representation learning, (2) \textbf{multi-scale causal kernel} modeling temporal dynamics across different time scales, and (3) \textbf{top-$k$ sparsification} for biologically plausible sparse connectivity structure.

\subsubsection{Mamba Encoder}

To encode neural activity time series for each ROI, we employ a standard Mamba encoder that is shared across all ROIs.
The standard Mamba encoder, based on efficient state-space model implementation, captures long-range temporal dependencies with linear complexity, enabling effective modeling of long time series.

For each ROI $i \in \{1, \ldots, R\}$ with neural activity time series $x_i \in \mathbb{R}^{T}$ (where $T$ is the temporal length), the encoding process proceeds as follows:
\begin{enumerate}
    \item \textbf{Input projection}: $x_i^{\mathrm{proj}} = \mathrm{Linear}(x_i) \in \mathbb{R}^{T \times H}$ where $H$ is the hidden dimension
    \item \textbf{Mamba-based encoding}: 
    \begin{equation}
        h_i = \mathrm{Mamba}(x_i^{\mathrm{proj}}) \in \mathbb{R}^{T \times H}
    \end{equation}
    \item \textbf{Temporal pooling} (optional): $e_i = \mathrm{AdaptiveAvgPool1d}(h_i) \in \mathbb{R}^{H}$ when temporal dimension is not preserved
\end{enumerate}

The standard Mamba encoder, as an efficient implementation of state-space models, captures long-range temporal dependencies with linear complexity.
By applying the same encoder to all ROIs, we learn shared temporal patterns while allowing ROI-specific characteristics to be captured in subsequent structure learning stages.

\subsubsection{Multi-Scale Causal Kernel}

Neural signal transmission occurs at various time scales due to synaptic delays, axonal conduction velocities, and indirect pathway transmission.
To explicitly model these multi-scale temporal dynamics, we introduce a multi-scale VAR-like causal kernel that simultaneously considers short-term, mid-term, and long-term temporal dependencies.

The multi-scale causal dynamics are modeled as:
\begin{equation}
    \mathbf{x}_{t+1} = \sum_{g \in \mathcal{G}} \sum_{\ell \in \mathcal{L}_g} \mathbf{A}^{(g)}_\ell \mathbf{x}_{t-\ell} + \boldsymbol{\epsilon}_t
    \label{eq:multiscale_var}
\end{equation}
where $\mathcal{G} = \{\mathrm{short}, \mathrm{mid}, \mathrm{long}\}$ denotes temporal scale groups, each with a corresponding lag set $\mathcal{L}_g$.
For fMRI data with TR = 2.0 seconds:
\begin{itemize}
    \item \textbf{Short-term}: $\mathcal{L}_{\mathrm{short}} = \{1, 2\}$ TRs (2--4 seconds) --- capturing direct neural connections and immediate causal influences
    \item \textbf{Mid-term}: $\mathcal{L}_{\mathrm{mid}} = \{3, 4, 5\}$ TRs (6--10 seconds) --- modeling intermediate pathway transmission and local network dynamics
    \item \textbf{Long-term}: $\mathcal{L}_{\mathrm{long}} = \{6, 7, \ldots, 10\}$ TRs (12--20 seconds) --- capturing indirect pathways, feedback loops, and long-range interactions
\end{itemize}

Each group $g$ has learnable kernel matrices $\mathbf{A}^{(g)}_\ell \in \mathbb{R}^{R \times R}$ for each lag $\ell \in \mathcal{L}_g$, where $A^{(g)}_{\ell,ij}$ represents the causal influence strength from ROI $i$ to ROI $j$ at time lag $\ell$ within scale group $g$.

The final causal strength matrix $\mathbf{C} \in \mathbb{R}^{R \times R}$ is computed as a weighted aggregation across all temporal scales:
\begin{equation}
    C_{ij} = \sum_{g \in \mathcal{G}} w_g \sum_{\ell \in \mathcal{L}_g} |A^{(g)}_{\ell, ij}|
    \label{eq:causal_strength}
\end{equation}
where $w_g$ are group-specific weights that can emphasize different temporal scales.
The absolute value operation $|\cdot|$ ensures that causal strength is represented as a positive scalar, aggregating influences across all time lags within each scale group.

\textbf{Integration with Dense Graph.}
The multi-scale causal kernel is integrated with the dense graph learned from ROI embeddings through a weighted combination:
\begin{equation}
    \mathbf{G}_{\mathrm{combined}} = (1-\alpha) \mathbf{G}_{\mathrm{dense}} + \alpha \mathbf{C}_{\mathrm{scaled}}
\end{equation}
where $\mathbf{G}_{\mathrm{dense}}$ is the pairwise MLP output from ROI embeddings, $\mathbf{C}_{\mathrm{scaled}}$ is the multi-scale causal matrix scaled to match the magnitude of $\mathbf{G}_{\mathrm{dense}}$, and $\alpha = 0.1$ is a small mixing coefficient that allows gradual integration of temporal dynamics.

\subsubsection{Top-$K$ Sparsification}

Real brain networks are extremely sparse, with empirical estimates suggesting that less than 10--15\% of possible connections represent actual causal relationships.
Standard $\ell_1$ regularization or threshold-based sparsification methods suffer from gradient vanishing under extreme sparsity and cannot guarantee exact target sparsity ratios.

We apply a top-$k$ selection mechanism that guarantees exact sparsity while maintaining gradient flow through selected connections.
The sparsification process consists of three steps:

\textbf{1. Dense Graph Generation.}
From ROI embeddings $\mathbf{E} \in \mathbb{R}^{B \times R \times H}$, we compute a dense causal graph using pairwise MLP:
\begin{equation}
    \mathbf{G}_{\mathrm{dense}} = \mathrm{PairwiseMLP}(\mathbf{E}) \in \mathbb{R}^{B \times R \times R}
\end{equation}
The dense graph is then passed through a learnable activation function that adapts to dataset statistics:
\begin{equation}
    \mathbf{G}_{\mathrm{scaled}} = \mathrm{CouplingActivation}(\mathbf{G}_{\mathrm{dense}})
\end{equation}
where $\mathrm{CouplingActivation}$ applies dataset-aware scaling to prevent scale collapse and ensures outputs are in a normalized space suitable for training.

\textbf{2. Top-$K$ Selection.}
For each sample in the batch, we select the top $k$ strongest connections globally:
\begin{equation}
    k = \lfloor \rho \cdot R^2 \rfloor
\end{equation}
where $\rho \in [0, 1]$ is the target sparsity ratio (e.g., $\rho = 0.1$ for 10\% connectivity).
The selection mask is constructed as:
\begin{equation}
    M_{ij}^{(b)} = \begin{cases}
        1 & \text{if } (i,j) \in \mathrm{TopK}(|\mathbf{G}_{\mathrm{scaled}}^{(b)}|, k) \\
        0 & \text{otherwise}
    \end{cases}
\end{equation}
where $\mathrm{TopK}(\cdot, k)$ returns the indices of the $k$ largest entries by absolute value.

\textbf{3. Sparse Graph Generation.}
The final sparse causal graph is obtained by element-wise multiplication:
\begin{equation}
    \mathbf{G}_{\mathrm{sparse}} = \mathbf{G}_{\mathrm{scaled}} \odot \mathbf{M}
    \label{eq:topk_sparse}
\end{equation}

This approach ensures: (i) \textbf{exact target sparsity ratio} $\rho$ is maintained, (ii) \textbf{gradient flow} is preserved through the selected $k$ edges, and (iii) \textbf{robustness} against overfitting to the overwhelming majority of null connections.

\textbf{Dynamic Sparsity Scheduling.}
During training, we employ a dynamic sparsity schedule that starts with a relatively dense graph and gradually increases sparsity:
\begin{equation}
    \rho(t) = \rho_{\mathrm{start}} - (\rho_{\mathrm{start}} - \rho_{\mathrm{end}}) \cdot \min\left(\frac{t}{T_{\mathrm{schedule}}}, 1\right)
\end{equation}
where $\rho_{\mathrm{start}} = 0.35$ (35\% connections), $\rho_{\mathrm{end}} = 0.25$ (25\% connections), and $T_{\mathrm{schedule}} = 50$ epochs.
This warmup strategy facilitates easier learning in early training while converging to biologically plausible sparse structures in later stages.

\subsection{Training Objectives}
\label{subsec:optimization}

Training causal models on biological neural data is challenging due to extreme sparsity, non-stationarity, and measurement noise.
Our objective function combines predictive accuracy with biological constraints:
\begin{equation}
    \mathcal{L}_{\mathrm{total}} = \mathcal{L}_{\mathrm{MAE}} + \lambda_{\mathrm{asym}} \mathcal{L}_{\mathrm{asym}} + \lambda_{\mathrm{scale}} \mathcal{L}_{\mathrm{scale}} + \lambda_{\mathrm{stab}} \mathcal{L}_{\mathrm{stab}}
\end{equation}

\textbf{1. Top-$K$ Weighted MAE Loss.}
To focus learning on the most informative connections, we adopt a weighted MAE loss over the top-$k$ strongest edges:
\begin{equation}
    \mathcal{L}_{\mathrm{MAE}} = \frac{1}{|\Omega_k|} \sum_{(i,j) \in \Omega_k} |A_{ij} - \widehat{A}_{ij}|
\end{equation}
where $\Omega_k$ denotes indices of the top-$k$ largest entries in the ground-truth adjacency matrix.
This acts as a robust estimator that prevents overfitting to the overwhelming majority of null connections.

\textbf{2. Asymmetry Loss.}
To preserve directionality ($A \to B$ vs $B \to A$), we penalize deviations in the asymmetric component:
\begin{equation}
    \mathcal{L}_{\mathrm{asym}} = \| (\widehat{\mathbf{A}} - \widehat{\mathbf{A}}^\top) - (\mathbf{A}^* - \mathbf{A}^{*\top}) \|_1
\end{equation}

\textbf{3. Scale Loss.}
To prevent output distribution collapse, we anchor the predicted statistics to target values:
\begin{equation}
    \mathcal{L}_{\mathrm{scale}} = |\mu_{\widehat{A}} - \mu^*| + |\sigma_{\widehat{A}} - \sigma^*|
\end{equation}

\textbf{4. Stability Loss.}
To ensure dynamical stability (spectral radius $< 1$), we impose a soft constraint via spectral norm:
\begin{equation}
    \mathcal{L}_{\mathrm{stab}} = \max\left(0, \|\widehat{\mathbf{A}}\|_2 - 1 \right)^2
\end{equation}

The hyperparameters $\lambda_{\mathrm{asym}}, \lambda_{\mathrm{scale}}, \lambda_{\mathrm{stab}}$ control the trade-off between predictive accuracy and biological plausibility.
We train the entire pipeline end-to-end, allowing the causal objective to guide earlier stages toward representations most suitable for causal discovery.

\section{Experiment Settings}
\label{appendix:experiment_settings}

This section summarizes all hyperparameters used to train \textbf{\textsc{\methodname}} across modalities and stages.

\subsection{fMRI Settings (TVB\_realistic\_nonstationary\_10k)}
\label{app:exp_fmri}

\textbf{Data \& preprocessing.}
All fMRI stages use \texttt{data/TVB\_realistic\_nonstationary\_10k} with BOLD normalization enabled.

\textbf{Core architecture and optimization.}
Table~\ref{tab:fmri_core} reports the backbone, dimensionality, and optimization hyperparameters across stages.

\begin{table}[H]
\centering
\small
\begin{tabular}{lccc}
\toprule
\textbf{Category} & \textbf{Stage 1 (BOLD$\rightarrow$Neural)} & \textbf{Stage 2 (Neural$\rightarrow$Causal)} & \textbf{Stage 3 (Joint)}\\
\midrule
Backbone & Mamba & Mamba & Mamba \\
Hidden size & 128 & 128 & 128 \\
\# ROIs & 68 & 68 & 68 \\
Causality option & --- & \texttt{sparse\_top\_k} & \texttt{sparse\_top\_k} \\
\midrule
Batch size & 16 & 32 & 32 \\
Epochs & 50 & 50 & 50 \\
Learning rate & $1\times 10^{-4}$ & $7.5\times 10^{-5}$ & $1\times 10^{-3}$ \\
Seed & 1 & 1 & 1 \\
\bottomrule
\end{tabular}
\caption{fMRI: core architecture and optimization settings.}
\label{tab:fmri_core}
\end{table}

\textbf{Loss weights and regularization.}
Table~\ref{tab:fmri_losses_regs} summarizes the loss weights and regularization coefficients used in each stage.

\begin{table}[H]
\centering
\small
\begin{tabular}{lccc}
\toprule
\textbf{Loss / regularizer} & \textbf{Stage 1} & \textbf{Stage 2} & \textbf{Stage 3}\\
\midrule
BOLD reconstruction weight & 0.5 & --- & 0.5 \\
Coupling loss weight & --- & 1.0 & 3.0 \\
Delay loss weight & --- & 30.0 & 30.0 \\
\midrule
$\lambda_{\mathrm{sparse}}$ & --- & $1\times 10^{-6}$ & 0.01 \\
$\lambda_{\mathrm{stab}}$ & --- & $5\times 10^{-3}$ & $5\times 10^{-3}$ \\
$\lambda_{\mathrm{sc}}$ & --- & 0.0 & 0.0 \\
$\lambda_{\mathrm{asym}}$ & --- & 1.0 & 1.0 \\
$\lambda_{\mathrm{scale}}$ & --- & 0.0 & 0.0 \\
$\lambda_{\mathrm{lap}}$ & --- & $5\times 10^{-4}$ & $5\times 10^{-4}$ \\
\midrule
$\lambda_{\mathrm{diversity}}$ & --- & --- & 10.0 \\
Diversity target std & --- & --- & 0.05 \\
$\lambda_{\mathrm{scale\_match}}$ & --- & --- & 1.0 \\
$\lambda_{\mathrm{recall}}$ & --- & --- & 0.1 \\
\bottomrule
\end{tabular}
\caption{fMRI: loss weights and regularization hyperparameters.}
\label{tab:fmri_losses_regs}
\end{table}

\textbf{Stage 2--3 details (selection, schedules, statistics, toggles).}
Table~\ref{tab:fmri_stage23_details} reports the remaining hyperparameters that are specific to the causal module and its training schedule.

\begin{longtable}{p{0.55\linewidth}cc}
\caption{fMRI: Stage 2--3 selection/schedules/statistics and architectural toggles.}
\label{tab:fmri_stage23_details}\\

\toprule
\textbf{Category} & \textbf{Stage 2} & \textbf{Stage 3} \\
\midrule
\endfirsthead

\multicolumn{3}{l}{\small\textit{Table \thetable\ (continued).}}\\
\toprule
\textbf{Category} & \textbf{Stage 2} & \textbf{Stage 3} \\
\midrule
\endhead

\midrule
\multicolumn{3}{r}{\small\textit{Continued on next page.}}\\
\endfoot

\bottomrule
\endlastfoot

\multicolumn{3}{l}{\textbf{Sparsity / selection}}\\
Sparsity ratio & 0.15 & --- \\
MAE Top-K ratio & 0.25 & 0.35 \\
Edge weight & 1.5 & 1.5 \\
Edge threshold & 0.015 & 0.015 \\
Edge percentile & 0.9 & 0.9 \\
Scale weight & 0.1 & 0.1 \\
\addlinespace[0.35em]

\multicolumn{3}{l}{\textbf{Scheduling}}\\
Regularization warmup epochs & 3 & 3 \\
MAE Top-K start ratio & 0.15 & 0.15 \\
MAE Top-K end ratio & 0.05 & 0.05 \\
MAE Top-K schedule length & 10 epochs & 10 epochs \\
\addlinespace[0.35em]

\multicolumn{3}{l}{\textbf{Coupling statistics}}\\
Auto coupling statistics & Enabled & Enabled \\
Auto statistics samples & 512 & 512 \\
Coupling norm & 1.0 & 1.0 \\
\addlinespace[0.35em]

\multicolumn{3}{l}{\textbf{Architecture toggles}}\\
Multi-scale & Enabled & Enabled \\
Groups mode & \texttt{default} & \texttt{default} \\
Laplacian & Enabled & Enabled \\
Downsample factor & 1 & 1 \\
Efficient GAT & Disabled & Disabled \\
Preserve temporal & Disabled & Disabled \\
\end{longtable}

\textbf{Notes (fMRI).}
As shown in Table~\ref{tab:fmri_losses_regs}, Stage~3 increases the coupling weight (1.0$\rightarrow$3.0) and sparsity regularization
($10^{-6}\rightarrow 10^{-2}$), and introduces diversity, scale-matching, and recall regularization to stabilize joint training.

\subsection{EEG Settings (nonstationary realistic datasets)}
\label{app:exp_eeg}

\textbf{Data \& preprocessing.}
EEG stages use \texttt{data\_10k\_nonstationary\_realistic}. ENIGMA structural connectivity is enabled.

\textbf{Stage-wise core settings.}
Table~\ref{tab:eeg_core} summarizes architecture, training, and checkpointing across EEG stages.

\begin{table}[H]
\centering
\small
\begin{tabular}{lccc}
\toprule
\textbf{Category} & \textbf{Stage 1 (DeepSIF pretrain)} & \textbf{Stage 2 (Causal mapping)} & \textbf{Stage 5 (Joint finetune)}\\
\midrule
Architecture & DeepSIF & Mamba & DeepSIF + Causal Mapper (Mamba)\\
\# ROIs & 68 & 68 & 68 \\
Preserve temporal & --- & 0 & 0 \\
ENIGMA SC & Enabled & Enabled & Enabled \\
\midrule
Batch size & 32 & 32 & 16 \\
Epochs & 50 & 50 & 20 \\
Learning rate & $5\times 10^{-4}$ & $7.5\times 10^{-5}$ & $1\times 10^{-5}$ \\
Seed & 2 & 2 & 1 \\
\midrule
Checkpoints & --- & --- & Stage 1 + Stage 2 pretrained\\
\bottomrule
\end{tabular}
\caption{EEG: stage-wise core settings.}
\label{tab:eeg_core}
\end{table}

\textbf{Stage 2 and Stage 5 details.}
Table~\ref{tab:eeg_stage2_stage5_details} reports regularization, edge filtering, scheduling, statistics, and architecture toggles for the EEG causal module.

\begin{longtable}{p{0.55\linewidth}cc}
\caption{EEG: Stage 2 and Stage 5 regularization, edge filtering, schedules, and toggles.}
\label{tab:eeg_stage2_stage5_details}\\

\toprule
\textbf{Category} & \textbf{Stage 2} & \textbf{Stage 5} \\
\midrule
\endfirsthead

\multicolumn{3}{l}{\small\textit{Table \thetable\ (continued).}}\\
\toprule
\textbf{Category} & \textbf{Stage 2} & \textbf{Stage 5} \\
\midrule
\endhead

\midrule
\multicolumn{3}{r}{\small\textit{Continued on next page.}}\\
\endfoot

\bottomrule
\endlastfoot

\multicolumn{3}{l}{\textbf{Regularization hyperparameters}}\\
$\lambda_{\mathrm{sparse}}$ & $1\times 10^{-6}$ & $1\times 10^{-6}$ \\
$\lambda_{\mathrm{stab}}$ & $5\times 10^{-3}$ & $5\times 10^{-3}$ \\
$\lambda_{\mathrm{sc}}$ & 0.0 & 0.0 \\
$\lambda_{\mathrm{asym}}$ & 1.0 & 1.0 \\
$\lambda_{\mathrm{lap}}$ & $5\times 10^{-4}$ & $5\times 10^{-4}$ \\
\addlinespace[0.35em]

\multicolumn{3}{l}{\textbf{Edge filtering}}\\
Edge weight & 1.5 & 1.5 \\
Edge threshold & 0.015 & 0.015 \\
Edge percentile & 0.9 & 0.9 \\
Scale weight & 0.002 & 0.002 \\
\addlinespace[0.35em]

\multicolumn{3}{l}{\textbf{Scheduling / MAE Top-K}}\\
MAE Top-K start ratio & 0.35 & 0.35 \\
MAE Top-K end ratio & 0.25 & 0.25 \\
MAE Top-K schedule length & 50 epochs & 20 epochs \\
Regularization warmup epochs & 40 & 20 \\
\addlinespace[0.35em]

\multicolumn{3}{l}{\textbf{Coupling statistics}}\\
Auto coupling statistics & Enabled & Enabled \\
Auto statistics samples & 512 & 512 \\
\addlinespace[0.35em]

\multicolumn{3}{l}{\textbf{Architecture toggles}}\\
Multi-scale & Enabled & Enabled \\
Groups mode & \texttt{default} & \texttt{checkpoint\_compat} \\
Laplacian & Enabled & Enabled \\
\end{longtable}

\textbf{Common EEG settings.}
WandB logging is enabled; ENIGMA structural connectivity is enabled across EEG stages.

\subsection{Compute / Server Details}
\label{app:server_details}

The compute environment is summarized in Table~\ref{tab:server_specs}.

\begin{table}[H]
\centering
\small
\begin{tabular}{ll}
\toprule
\textbf{Component} & \textbf{Specification} \\
\midrule
GPU & 7$\times$ NVIDIA GeForce RTX 3090 \\
GPU memory & 24{,}576 MiB per GPU ($\approx$24 GB), $\approx$168 GB total \\
Driver / CUDA & Driver 550.90.07; CUDA 12.4 (driver), 12.5 (compiler) \\
\midrule
CPU & Intel Xeon Gold 6240R @ 2.40 GHz \\
Cores / threads & 48 cores (24$\times$2 sockets), 96 threads (HT) \\
Max clock & 4.0 GHz \\
\midrule
RAM & 502 GB total (436 GB available) \\
Swap & 8 GB \\
\midrule
OS / Python & Ubuntu 20.04.6 LTS; Python 3.8.10 \\
\bottomrule
\end{tabular}
\caption{Compute environment used for all experiments.}
\label{tab:server_specs}
\end{table}

\newpage

\section{EEG Simulation Framework: Stationary and Non-stationary Data Generation}
\label{sec:appendix_EEG_simulator}

This appendix describes the physics-aware EEG simulation framework used to generate stationary and non-stationary datasets with biologically plausible characteristics. The simulator balances biological realism, controllability, and scalability, while providing ground-truth causal topology, time-varying coupling strengths, and transmission delays.

\subsection{At-a-glance summary}

\begin{table}[h!]
\centering
\small
\begin{tabular}{p{0.22\linewidth}p{0.72\linewidth}}
\toprule
\textbf{Component} & \textbf{Key choices} \\
\midrule
ROI parcellation & Desikan--Killiany atlas, $N=68$ ROIs \\
Structural connectivity & ENIGMA SC (healthy controls; MNI coordinates) \\
Effective connectivity & SC-to-EC sparsification + directionality + self-inhibition \\
Stationarity modes & Stationary ($\sigma_{\text{drift}}=0$) vs.\ non-stationary AR(1) drift ($\rho=0.90$, $\sigma_{\text{drift}}=0.10$) \\
Delays & Distance-based: $\tau_{ij}=d_{ij}/v+\tau_{\mathrm{syn}}$, with clipping (5--20 ms) \\
Neural dynamics & Neural mass with band assignments (alpha-dominant) + synaptic time constants \\
Forward model & Leadfield $L\in\mathbb{R}^{129\times N}$ (GSN-128 + Cz) + low-rank perturbation \\
Sensor-space shaping & Band-weighted reconstruction + correlation control + RMS normalization + optional artifacts \\
Downsampling & 500 Hz $\rightarrow$ 250 Hz via FIR anti-aliasing (\texttt{scipy.signal.decimate}, factor 2) \\
Outputs & LFP $X$, topology $M$, coupling $B$ or $B_t$, delays $\Tau$, EEG, leadfield, metadata \\
\bottomrule
\end{tabular}
\caption{At-a-glance summary of the EEG simulation framework.}
\label{tab:eeg_sim_at_a_glance}
\end{table}

\subsection{Stationary vs.\ non-stationary coupling dynamics}
\label{app:eeg_stationary_nonstationary}

We generate two EEG dataset variants with identical biological constraints but different coupling dynamics (Table~\ref{tab:eeg_stationary_nonstationary}).

\begin{table}[h!]
\centering
\small
\begin{tabular}{p{0.36\linewidth}p{0.28\linewidth}p{0.28\linewidth}}
\toprule
\textbf{Property} & \textbf{Stationary} & \textbf{Non-stationary} \\
\midrule
Coupling matrix & Fixed $\mathbf{B}\in\mathbb{R}^{N\times N}$ & Time-varying $\mathbf{B}_t\in\mathbb{R}^{T\times N\times N}$ \\
Topology mask & Implicit in $\mathbf{B}$ & Fixed $\mathbf{M}\in\{0,1\}^{N\times N}$ \\
Construction & --- & $\mathbf{B}_t=\mathbf{M}\odot \mathbf{W}_t$ \\
Drift process & $\sigma_{\text{drift}}=0.0$ & AR(1): $\rho=0.90$, $\sigma_{\text{drift}}=0.10$ \\
\bottomrule
\end{tabular}
\caption{Stationary vs.\ non-stationary dataset generation.}
\label{tab:eeg_stationary_nonstationary}
\end{table}

\subsection{Structural connectivity and causal topology}
\label{app:eeg_topology}

\textbf{ENIGMA structural connectivity.}
We leverage ENIGMA structural connectivity (SC) under the Desikan--Killiany atlas ($N=68$ ROIs), providing biologically grounded connectivity and ROI coordinates in MNI space.

\textbf{SC-to-effective connectivity conversion.}
Because SC is dense, we convert it to sparse effective connectivity using the procedure summarized in Table~\ref{tab:eeg_sc_to_ec}.

\begin{table}[h!]
\centering
\small
\begin{tabular}{p{0.30\linewidth}p{0.62\linewidth}}
\toprule
\textbf{Step} & \textbf{Specification} \\
\midrule
Weak connection filtering & Remove bottom 50\% weakest SC edges \\
Probabilistic sampling & Retain $(i,j)$ with probability $\min\!\left(1, \mathrm{SC}_{\text{norm}}(i,j)\cdot r_{\text{retention}}\right)$, $r_{\text{retention}}=1.5$ \\
Target density & 10--15\% off-diagonal density (effective connectivity) \\
Directional bias & 63\% feedforward-only, 30\% bidirectional, 7\% feedback-only \\
Self-inhibition & Add self-connection with probability $p_{\text{self}}=0.25$ per ROI \\
\bottomrule
\end{tabular}
\caption{Procedure for converting dense ENIGMA SC into sparse effective connectivity.}
\label{tab:eeg_sc_to_ec}
\end{table}

\subsection{Distance-based transmission delays}
\label{app:eeg_delays}

Each directed edge is assigned a transmission delay
\[
\tau_{ij}=\frac{d_{ij}}{v}+\tau_{\mathrm{syn}},
\]
with parameters summarized in Table~\ref{tab:eeg_delays}.

\begin{table}[h!]
\centering
\small
\begin{tabular}{p{0.38\linewidth}p{0.54\linewidth}}
\toprule
\textbf{Quantity} & \textbf{Specification} \\
\midrule
Distance $d_{ij}$ & Euclidean distance between ROI centroids (mm; MNI coordinates) \\
Conduction velocity $v$ & $v\sim\mathcal{U}(6.0,12.0)$ m/s \\
Synaptic delay $\tau_{\mathrm{syn}}$ & $\tau_{\mathrm{syn}}\sim\mathcal{U}(0.001,0.003)$ s (1--3 ms) \\
Clipping & Delays clipped to 5--20 ms \\
Self-delays & Self-connections: 5--10 ms (distance-independent) \\
\bottomrule
\end{tabular}
\caption{Delay model and parameter ranges.}
\label{tab:eeg_delays}
\end{table}

\subsection{Neural mass model configuration}
\label{app:eeg_nmm}

\textbf{Frequency band assignments.}
Each ROI is assigned a dominant band with probability
\[
P(\text{band})=[0.10,0.15,0.70,0.05,0.0] \text{ for [Delta, Theta, Alpha, Beta, Gamma]}.
\]
Table~\ref{tab:eeg_bands_tau} summarizes the spectral and synaptic parameters.

\begin{table}[h!]
\centering
\small
\begin{tabular}{p{0.38\linewidth}p{0.54\linewidth}}
\toprule
\textbf{Parameter} & \textbf{Specification} \\
\midrule
Band distribution & Delta 10\%, Theta 15\%, Alpha 70\%, Beta 5\%, Gamma 0\% \\
Alpha band & 8--13 Hz (dominant in resting-state mimicry) \\
Excitatory time constant $\tau_E$ & $\tau_E\sim\mathcal{U}(0.010,0.015)$ s (10--15 ms) \\
Inhibitory time constant $\tau_I$ & $\tau_I\sim\mathcal{U}(0.060,0.100)$ s (60--100 ms) \\
Sigmoid parameters & threshold $=0.5$, slope $=2.0$ \\
\bottomrule
\end{tabular}
\caption{Neural mass model spectral and synaptic parameters.}
\label{tab:eeg_bands_tau}
\end{table}

\subsection{Time-varying coupling weights (non-stationary only)}
\label{app:eeg_drift}

For non-stationary datasets, coupling weights evolve via AR(1) drift:
\[
\delta_t^{(i,j)}=\rho\,\delta_{t-1}^{(i,j)}+\epsilon_t, \quad
w_t^{(i,j)}=w_0^{(i,j)}+\delta_t^{(i,j)}, \quad
\epsilon_t\sim\mathcal{N}(0,\sigma_{\text{drift}}^2).
\]
Key choices are summarized in Table~\ref{tab:eeg_drift}.

\begin{table}[h!]
\centering
\small
\begin{tabular}{p{0.38\linewidth}p{0.54\linewidth}}
\toprule
\textbf{Quantity} & \textbf{Specification} \\
\midrule
AR(1) coefficient & $\rho=0.90$ \\
Drift noise & $\sigma_{\text{drift}}=0.10$ \\
Initial weights $w_0^{(i,j)}$ & From normalized SC, clipped to $[0.03,0.3]$ \\
Sign assignment & Self: inhibitory (negative); inter-ROI: 80\% excitatory, 20\% inhibitory \\
\bottomrule
\end{tabular}
\caption{Non-stationary drift process and coupling initialization.}
\label{tab:eeg_drift}
\end{table}

\subsection{LFP generation}
\label{app:eeg_lfp}

\textbf{Latent activity constraints.}
LFP generation enforces: (i) delayed coupling via $\Tau$, (ii) band-specific oscillations per ROI,
(iii) slow amplitude modulation (0.1--0.5 Hz), and (iv) signal mixing (30\% original, 70\% clean oscillation).

\subsection{Leadfield matrix and scalp EEG generation}
\label{app:eeg_forward_model}

\textbf{Leadfield matrix.}
We generate a leadfield $\mathbf{L}\in\mathbb{R}^{129\times N}$ using the HydroCel GSN-128 layout (129 channels including Cz reference) with standard 3D electrode coordinates. To model inter-individual variability, we apply a low-rank perturbation (magnitude 0.8--1.5, rank 5--10).

\textbf{Scalp projection.}
ROI activity $\mathbf{X}(t)$ is projected to sensor space as
\[
\mathbf{y}_{\text{raw}}(t)=\mathbf{X}(t)\mathbf{L}^\top.
\]

\textbf{Sensor-space shaping.}
We match empirical resting-state statistics by combining band-filtered components from $\mathbf{y}_{\text{raw}}$ with channel-specific synthetic oscillations. The target band mixture and shaping steps are summarized in Table~\ref{tab:eeg_sensor_shaping}.

\begin{longtable}{p{0.30\linewidth}p{0.62\linewidth}}
\caption{EEG sensor-space shaping: target band mixture, correlation control, normalization, and artifacts.}
\label{tab:eeg_sensor_shaping}\\
\toprule
\textbf{Component} & \textbf{Specification} \\
\midrule
\endfirsthead
\multicolumn{2}{l}{\small\textit{Table \thetable\ (continued).}}\\
\toprule
\textbf{Component} & \textbf{Specification} \\
\midrule
\endhead
\midrule
\multicolumn{2}{r}{\small\textit{Continued on next page.}}\\
\endfoot
\bottomrule
\endlastfoot

Target band mixture & Delta 10\%, Theta 12\%, \textbf{Alpha 70\%}, Beta 6\%, Gamma 2\% \\
Band extraction & Extract each band from raw EEG to preserve spatial correlation structure \\
Synthetic oscillations & Generate independent oscillations per channel (100\% channel-specific) \\
Band weighting & Alpha: $\text{raw\_weight}=3.0$, $\text{synth\_weight}=2.5$; other bands use lower weights \\
Correlation target & Inter-channel correlation 0.2--0.4 \\
Decorrelation strategy & Add independent $1/f$ noise per channel; minimize shared noise; avoid average reference \\
RMS normalization & $\text{target\_RMS}\sim\mathcal{U}(30,50)$ $\mu$V per channel; soft-clip outliers $>80$ $\mu$V via \texttt{tanh} \\
Optional artifacts: blinks & 10--20/min, duration $\sim$300 ms, stronger in frontal channels \\
Optional artifacts: muscle & 20--60 Hz, duration 100--500 ms, stronger in temporal/peripheral channels \\
Optional artifacts: $1/f$ & Spectral slope $1/f^{1.2}$; 95\% shared, 5\% independent per channel \\
\end{longtable}

\subsection{Downsampling with anti-aliasing}
\label{app:eeg_downsampling}

Signals are downsampled using \texttt{scipy.signal.decimate} with FIR filtering (Table~\ref{tab:eeg_downsampling}).

\begin{table}[h!]
\centering
\small
\begin{tabular}{p{0.40\linewidth}p{0.52\linewidth}}
\toprule
\textbf{Quantity} & \textbf{Specification} \\
\midrule
Original sampling rate & 500 Hz \\
Downsampling factor & 2 (default) \\
Effective sampling rate & 250 Hz \\
Method & FIR anti-aliasing via \texttt{scipy.signal.decimate} \\
\bottomrule
\end{tabular}
\caption{Downsampling configuration.}
\label{tab:eeg_downsampling}
\end{table}

\subsection{Validation checks}
\label{app:eeg_validation}

The first subject (\texttt{subject\_0000}) is automatically validated against biological plausibility checks summarized in Table~\ref{tab:eeg_validation}.

\begin{table}[h!]
\centering
\small
\begin{tabular}{p{0.18\linewidth}p{0.33\linewidth}p{0.39\linewidth}}
\toprule
\textbf{Object} & \textbf{Metric} & \textbf{Target range / criterion} \\
\midrule
EEG & RMS amplitude & 10--100 $\mu$V \\
EEG & Frequency distribution & Alpha band 30--50\% (target: 70\%) \\
EEG & Spatial correlation & 0.2--0.5 \\
EEG & Kurtosis & 3--8 (artifact-driven) \\
\midrule
LFP & Std.\ deviation & 0.2--0.8 \\
LFP & Oscillatory content & Clear frequency bands present \\
LFP & Connectivity-driven dynamics & Consistent with coupling matrix \\
\midrule
Leadfield & Rank & $\ge N_{\text{ROIs}}$ \\
Leadfield & Condition number & Numerical stability check \\
\midrule
Topology & Sparsity & 10--20\% off-diagonal density \\
Topology & Asymmetry & Directional connections verified \\
Topology & Hierarchy & Feedforward/feedback ratio enforced \\
\bottomrule
\end{tabular}
\caption{Automatic plausibility validation for \texttt{subject\_0000}.}
\label{tab:eeg_validation}
\end{table}

\subsection{Output file structure}
\label{app:eeg_outputs}

Each subject directory (\texttt{subject\_XXXX/}) contains the files listed in Table~\ref{tab:eeg_outputs}.

\begin{table}[h!]
\centering
\small
\begin{tabular}{p{0.30\linewidth}p{0.60\linewidth}}
\toprule
\textbf{File} & \textbf{Contents (shape)} \\
\midrule
\texttt{X.npy} & LFP signal $(T, N)$ \\
\texttt{B\_t.npy} & Time-varying couplings $(T, N, N)$ [Non-stationary] \\
\texttt{B.npy} & Fixed coupling matrix $(N, N)$ [Stationary] \\
\texttt{M.npy} & Fixed topology mask $(N, N)$ \\
\texttt{Tau.npy} & Delay matrix $(N, N)$ \\
\texttt{EEG.npy} & Scalp EEG signal $(T, 129)$ \\
\texttt{leadfield.npy} & Leadfield matrix $(129, N)$ \\
\texttt{channel\_names.npy} & Channel names $(129,)$ \\
\texttt{sensor\_positions.npy} & Sensor positions $(129, 3)$ \\
\texttt{meta.json} & Metadata (sampling frequency, duration, parameters, etc.) \\
\bottomrule
\end{tabular}
\caption{Per-subject output files produced by the EEG simulator.}
\label{tab:eeg_outputs}
\end{table}

\subsection{Summary of biological plausibility constraints}
\label{app:eeg_bio_summary}

The simulation enforces the biological realism constraints listed in Table~\ref{tab:eeg_bio_summary}.

\begin{table}[h!]
\centering
\small
\begin{tabular}{p{0.36\linewidth}p{0.54\linewidth}}
\toprule
\textbf{Constraint} & \textbf{Implementation} \\
\midrule
Real brain connectivity & ENIGMA SC (healthy controls) \\
Distance-based delays & $v\in[6,12]$ m/s; synaptic delay 1--3 ms; clip 5--20 ms \\
Spectral realism & Alpha dominance (70\%) in ROI band assignments; sensor target mixture uses alpha 70\% \\
Electrode layout & HydroCel GSN-128 + Cz (129 channels) \\
Spatial correlations & Target inter-channel correlation 0.2--0.4 via noise-based decorrelation \\
Amplitude realism & RMS $\sim\mathcal{U}(30,50)$ $\mu$V with soft clipping \\
Artifacts & Optional eye blinks + muscle bursts + $1/f$ noise \\
$1/f$ spectrum & Slope $1/f^{1.2}$ (with shared/independent components) \\
E/I balance & 80\% excitatory, 20\% inhibitory inter-ROI; self inhibitory \\
Hierarchy & Feedforward/bidirectional/feedback = 63\% / 30\% / 7\% \\
\bottomrule
\end{tabular}
\caption{Key biological plausibility constraints enforced by the simulator.}
\label{tab:eeg_bio_summary}
\end{table}

These design choices ensure the generated EEG preserves key empirical characteristics while providing synthetic data with known causal structure for principled model training and evaluation.

\newpage

\section{fMRI Simulation Framework: Stationary and Non-stationary Data Generation}
\label{sec:appendix_fMRI_simulator}

This appendix describes the physics-aware fMRI simulation framework used to generate stationary and non-stationary fMRI datasets with biologically plausible characteristics. The simulator is based on The Virtual Brain (TVB) architecture and provides ground truth for causal topology, coupling strengths (fixed or time-varying), conduction delays, and neural activity prior to hemodynamic filtering.

\subsection{At-a-glance summary}

\begin{table}[h!]
\centering
\small
\begin{tabular}{p{0.42\linewidth}p{0.52\linewidth}}
\toprule
\textbf{Component} & \textbf{Key choices} \\
\midrule
ROI parcellation & Desikan--Killiany atlas, $N=68$ ROIs \\
Structural connectivity & ENIGMA SC (healthy controls; MNI coordinates) \\
Effective connectivity & SC-to-EC sparsification + directionality + self-inhibition \\
Stationarity modes & Stationary (fixed $\mathbf{B}$) vs.\ non-stationary AR(1) drift ($\rho=0.90$, $\sigma_{\text{drift}}=0.10$) \\
Delays & Distance-based: $\tau_{ij}=d_{ij}/v+\tau_{\text{syn}}$, clipped (2.5--50 ms) \\
Neural dynamics & Wilson--Cowan E/I neural mass model with external drive + delayed coupling \\
Hemodynamics & ROI-specific HRF (double-gamma), normalized and scaled to 1--5\% PSC \\
Simulation resolution & High-res simulation at 100 Hz ($dt=10$ ms) then downsample to TR (e.g.\ 2.0 s) \\
BOLD preprocessing & Smoothing, high-pass 0.008 Hz, optional low-pass 0.15 Hz, PSC conversion \\
Dataset scale & 10,000 subjects per condition; duration 480 s; TR 2.0 s; 240 timepoints \\
Outputs & BOLD, neural activity, $\mathbf{B}$ or $\mathbf{B}_t$, topology $\mathbf{M}$, delays $\boldsymbol{\tau}$, HRF params \\
\bottomrule
\end{tabular}
\caption{At-a-glance summary of the fMRI simulation framework.}
\label{tab:fmri_sim_at_a_glance}
\end{table}

\subsection{Stationary vs.\ non-stationary coupling dynamics}
\label{app:fmri_stationary_nonstationary}

We generate stationary and non-stationary fMRI datasets that share the same biological constraints but differ in coupling dynamics (Table~\ref{tab:fmri_stationary_nonstationary}).

\begin{table}[h!]
\centering
\small
\begin{tabular}{p{0.36\linewidth}p{0.28\linewidth}p{0.28\linewidth}}
\toprule
\textbf{Property} & \textbf{Stationary} & \textbf{Non-stationary} \\
\midrule
Coupling & Fixed $\mathbf{B}\in\mathbb{R}^{N\times N}$ & Time-varying $\mathbf{B}_t\in\mathbb{R}^{T\times N\times N}$ \\
Topology mask & Implicit in $\mathbf{B}$ & Fixed $\mathbf{M}\in\{0,1\}^{N\times N}$ \\
Construction & --- & $\mathbf{B}_t=\mathbf{M}\odot \mathbf{W}_t$ \\
Drift process & --- & AR(1): $\rho=0.90$, $\sigma_{\text{drift}}=0.10$ \\
Temporal resolution & TR resolution & Generated at 100 Hz, averaged to TR \\
\bottomrule
\end{tabular}
\caption{Stationary vs.\ non-stationary coupling generation.}
\label{tab:fmri_stationary_nonstationary}
\end{table}

\subsection{ROI-specific hemodynamic response function (HRF)}
\label{app:fmri_hrf}

\textbf{HRF heterogeneity.}
We incorporate ROI-specific HRF parameters to reflect regional variability in hemodynamic responses. Each ROI $i$ is assigned a double-gamma HRF
\[
h_i(t)=\frac{t^{a_1^{(i)}-1}e^{-t/b_1^{(i)}}}{b_1^{(i)a_1^{(i)}}\Gamma(a_1^{(i)})}-c^{(i)}\cdot
\frac{t^{a_2^{(i)}-1}e^{-t/b_2^{(i)}}}{b_2^{(i)a_2^{(i)}}\Gamma(a_2^{(i)})},
\]
where $(a_1^{(i)},b_1^{(i)})$ and $(a_2^{(i)},b_2^{(i)})$ are derived from peak and undershoot delays, respectively. HRFs are normalized to unit peak and scaled to yield realistic BOLD changes (1--5\% PSC).

Key HRF parameter ranges and grouping are summarized in Table~\ref{tab:fmri_hrf_params}.

\begin{table}[h!]
\centering
\small
\begin{tabular}{p{0.42\linewidth}p{0.48\linewidth}}
\toprule
\textbf{Quantity} & \textbf{Specification} \\
\midrule
Peak delay $t_{\text{peak}}^{(i)}$ & $[4.0, 9.0]$ s \\
Undershoot delay $t_{\text{undershoot}}^{(i)}$ & $[12.0, 22.0]$ s \\
Undershoot scale $c^{(i)}$ & $[0.15, 0.50]$ \\
ROI HRF groups & Fast (peak $\sim$5--6 s), Medium (6--7 s), Slow (6.5--8 s) \\
Normalization / scaling & Unit peak; scaled to 1--5\% PSC \\
\bottomrule
\end{tabular}
\caption{ROI-specific HRF parameterization (double-gamma).}
\label{tab:fmri_hrf_params}
\end{table}

\subsection{Structural connectivity and causal topology}
\label{app:fmri_topology}

\textbf{ENIGMA structural connectivity.}
We use ENIGMA structural connectivity (SC) under the Desikan--Killiany atlas ($N=68$ ROIs) to ground the topology in biologically realistic anatomy.

\textbf{SC-to-effective connectivity conversion.}
Because SC is dense, we convert it to sparse effective connectivity using the procedure summarized in Table~\ref{tab:fmri_sc_to_ec}.

\begin{table}[h!]
\centering
\small
\begin{tabular}{p{0.30\linewidth}p{0.62\linewidth}}
\toprule
\textbf{Step} & \textbf{Specification} \\
\midrule
Weak connection filtering & Remove bottom 25\% weakest SC edges \\
Probabilistic sampling & Retain $(i,j)$ with probability $\min\!\left(1, \mathrm{SC}_{\text{norm}}(i,j)\cdot r_{\text{retention}}\right)$, $r_{\text{retention}}=1.5$ \\
Target density & 10--15\% off-diagonal density \\
Directional bias & 45\% feedforward-only, 50\% bidirectional, 5\% feedback-only \\
Self-inhibition & Add self-connection with probability $p_{\text{self}}=0.25$ per ROI \\
\bottomrule
\end{tabular}
\caption{Procedure for converting ENIGMA SC into sparse effective connectivity for fMRI.}
\label{tab:fmri_sc_to_ec}
\end{table}

\subsection{Distance-based conduction delays}
\label{app:fmri_delays}

Transmission delays are computed from inter-regional distances:
\[
\tau_{ij}=\frac{d_{ij}}{v}+\tau_{\text{syn}}.
\]
Parameters are summarized in Table~\ref{tab:fmri_delays}.

\begin{table}[h!]
\centering
\small
\begin{tabular}{p{0.38\linewidth}p{0.54\linewidth}}
\toprule
\textbf{Quantity} & \textbf{Specification} \\
\midrule
Distance $d_{ij}$ & Euclidean distance between ROI centroids (mm; MNI coordinates) \\
Conduction velocity $v$ & $v\sim\mathcal{U}(4.0, 8.0)$ m/s \\
Synaptic delay $\tau_{\text{syn}}$ & $\tau_{\text{syn}}\sim\mathcal{U}(0.003, 0.008)$ s (3--8 ms) \\
Clipping & 2.5--50 ms \\
Self-delays & 3--8 ms (distance-independent) \\
\bottomrule
\end{tabular}
\caption{Distance-based conduction delay model.}
\label{tab:fmri_delays}
\end{table}

\subsection{Wilson--Cowan neural mass model}
\label{app:fmri_wc}

\textbf{Local dynamics.}
Latent neural activity is generated using a Wilson--Cowan neural mass model with excitatory/inhibitory subpopulations per ROI.
Local parameters and time constants are summarized in Table~\ref{tab:fmri_wc_params}.

\begin{table}[h!]
\centering
\small
\begin{tabular}{p{0.42\linewidth}p{0.48\linewidth}}
\toprule
\textbf{Quantity} & \textbf{Specification} \\
\midrule
$w_{EE}^{(i)}$ & $\mathcal{U}(1.2, 1.6)$ \\
$w_{EI}^{(i)}$ & $\mathcal{U}(0.8, 1.2)$ \\
$w_{IE}^{(i)}$ & $\mathcal{U}(1.0, 1.4)$ \\
$w_{II}^{(i)}$ & $\mathcal{U}(0.4, 0.8)$ \\
$\tau_E$ & 0.010 s (10 ms) \\
$\tau_I$ & 0.100 s (100 ms) \\
Activation & $\sigma(x)=1/(1+\exp(-2x))$ \\
\bottomrule
\end{tabular}
\caption{Wilson--Cowan local parameters and time constants.}
\label{tab:fmri_wc_params}
\end{table}

\textbf{Dynamics.}
For ROI $i$,
\begin{align}
\tau_E \dot{E}_i(t) &= -E_i(t) + \sigma\!\left(w_{EE}^{(i)} E_i - w_{EI}^{(i)} I_i + u_i(t)\right), \\
\tau_I \dot{I}_i(t) &= -I_i(t) + \sigma\!\left(w_{IE}^{(i)} E_i - w_{II}^{(i)} I_i\right).
\end{align}

\textbf{External input.}
The external drive combines pink noise, infraslow oscillations, and sparse events, with weights summarized in Table~\ref{tab:fmri_external_input}:
\[
u_i^{\text{external}}(t)=0.5\,\text{slow}_i(t)+0.3\,\text{pink}_i(t)+0.2\,\text{events}_i(t).
\]

\begin{table}[h!]
\centering
\small
\begin{tabular}{p{0.36\linewidth}p{0.54\linewidth}}
\toprule
\textbf{Component} & \textbf{Specification} \\
\midrule
Pink noise & Independent $1/f$ noise per ROI \\
Slow fluctuations & 0.01--0.04 Hz; sum of 8 sinusoids with random phases \\
Sparse events & Rate $\sim 0.08$/s; duration $\sim 150$ ms \\
Mixture weights & slow: 0.5, pink: 0.3, events: 0.2 \\
\bottomrule
\end{tabular}
\caption{External input components for Wilson--Cowan dynamics.}
\label{tab:fmri_external_input}
\end{table}

\textbf{Delayed coupling input.}
Coupled input is
\[
u_i^{\text{coupled}}(t)=\sum_{j=1}^{N}B_{ji}(t)\,E_j(t-\tau_{ji}),
\quad
u_i(t)=u_i^{\text{external}}(t)+0.05\,u_i^{\text{coupled}}(t).
\]
The coupling term is scaled (0.05) to maintain stability.

\subsection{High-resolution simulation, HRF convolution, and TR downsampling}
\label{app:fmri_pipeline}

Neural dynamics are simulated at 100 Hz ($dt=0.01$ s) prior to hemodynamic filtering and TR downsampling.
The end-to-end pipeline is summarized in Table~\ref{tab:fmri_pipeline_steps}.

\begin{table}[h!]
\centering
\small
\begin{tabular}{p{0.22\linewidth}p{0.70\linewidth}}
\toprule
\textbf{Step} & \textbf{Specification} \\
\midrule
Simulation & Euler integration with $dt=0.01$ s at 100 Hz; store $E_i(t)$ \\
Noise / smoothing & Add 2\% pink noise; Gaussian smoothing with $\sigma=0.3$ s \\
Hemodynamics & Convolve each ROI with its HRF kernel $h_i(t)$ \\
TR downsampling & Average over $T_{\text{TR}}/dt$ samples (e.g.\ TR=2.0 s) \\
\bottomrule
\end{tabular}
\caption{High-resolution simulation pipeline prior to TR sampling.}
\label{tab:fmri_pipeline_steps}
\end{table}

\subsection{Time-varying coupling weights (non-stationary only)}
\label{app:fmri_drift}

For non-stationary datasets, coupling weights evolve via AR(1) drift at 100 Hz:
\[
\delta_t^{(i,j)}=\rho\,\delta_{t-1}^{(i,j)}+\epsilon_t,\quad
w_t^{(i,j)}=w_0^{(i,j)}+\delta_t^{(i,j)},\quad
\epsilon_t\sim\mathcal{N}(0,\sigma_{\text{drift}}^2),
\]
with typical values $\rho=0.90$ and $\sigma_{\text{drift}}=0.10$.
Initialization and constraints are summarized in Table~\ref{tab:fmri_drift_params}.

\begin{table}[h!]
\centering
\small
\begin{tabular}{p{0.40\linewidth}p{0.52\linewidth}}
\toprule
\textbf{Quantity} & \textbf{Specification} \\
\midrule
AR(1) coefficient & $\rho=0.90$ \\
Drift noise & $\sigma_{\text{drift}}=0.10$ \\
Initialization & $w_0^{(i,j)}=0.1+1.4\cdot \dfrac{\mathrm{SC}_{\text{norm}}(i,j)-\mathrm{SC}_{\min}}{\mathrm{SC}_{\max}-\mathrm{SC}_{\min}}$ \\
Clipping & $w_t^{(i,j)}\in[0.1,1.5]$ at each timestep \\
Sign assignment & Self: inhibitory; inter-ROI: 85\% excitatory, 15\% inhibitory \\
Downsampling & Average $\mathbf{B}_t$ over each TR window to get TR-resolution $\mathbf{B}_t\in\mathbb{R}^{T\times N\times N}$ \\
\bottomrule
\end{tabular}
\caption{Non-stationary drift process and coupling constraints (fMRI).}
\label{tab:fmri_drift_params}
\end{table}

\subsection{BOLD preprocessing and PSC conversion}
\label{app:fmri_preproc}

To match empirical resting-state fMRI characteristics while preserving low-frequency content, we apply the preprocessing steps in Table~\ref{tab:fmri_preproc}.

\begin{table}[h!]
\centering
\small
\begin{tabular}{p{0.28\linewidth}p{0.64\linewidth}}
\toprule
\textbf{Step} & \textbf{Specification} \\
\midrule
Temporal smoothing & Gaussian smoothing with $\sigma=0.5$ TR \\
High-pass filtering & 2nd-order Butterworth, cutoff 0.008 Hz \\
Low-pass filtering (optional) & 2nd-order Butterworth, cutoff 0.15 Hz (below Nyquist for TR=2.0 s) \\
PSC conversion & Center each ROI; scale to target std.\ 20.0 (PSC units); ROI amplitude factor $\sim\mathcal{U}(0.8,1.2)$ \\
\bottomrule
\end{tabular}
\caption{BOLD preprocessing for realistic PSC-scale resting-state fMRI.}
\label{tab:fmri_preproc}
\end{table}

This preprocessing targets PSC ranges of roughly $\pm 3$--$8\%$ (std.\ $\sim 1.5$--$4\%$) while preserving the dominant 0.01--0.04 Hz band.

\subsection{Connectivity constraints and sparsity}
\label{app:fmri_constraints}

\textbf{E/I balance and sparsity.}
Connectivity is constrained to be sparse (10--15\% non-zero off-diagonal density). Inter-regional connections are 85\% excitatory and 15\% inhibitory; self-connections are always inhibitory.

\textbf{Distance-dependent connectivity (synthetic option).}
When ENIGMA is not used, synthetic edges follow
\[
P(i\!\leftrightarrow\! j)\propto \exp\!\left(-\frac{d_{ij}}{\lambda}\right),
\quad \lambda=0.5,
\]
with a $2\times$ boost for local connections ($d_{ij}<0.3$).

\subsection{Dataset generation protocol and outputs}
\label{app:fmri_protocol}

We generate 10{,}000 simulated subjects for each condition (stationary and non-stationary). Per-subject randomized components and global simulation settings are summarized in Table~\ref{tab:fmri_protocol_params}.

\begin{table}[h!]
\centering
\small
\begin{tabular}{p{0.40\linewidth}p{0.52\linewidth}}
\toprule
\textbf{Quantity} & \textbf{Specification} \\
\midrule
Per-subject randomized components & $\mathbf{M}$, ROI HRFs, local WC parameters, delays $\boldsymbol{\tau}$; plus drift params for non-stationary \\
Duration & 480 s (8 min) \\
TR & 2.0 s \\
\# timepoints & 240 \\
\# ROIs & 68 \\
\midrule
Stored outputs & $\mathbf{X}_{\text{BOLD}}\in\mathbb{R}^{68\times 240}$; $\mathbf{X}_{\text{neural}}\in\mathbb{R}^{68\times 240}$; $\mathbf{B}$ or $\mathbf{B}_t\in\mathbb{R}^{240\times 68\times 68}$; $\mathbf{M}$; $\boldsymbol{\tau}$; per-ROI HRF params \\
\bottomrule
\end{tabular}
\caption{Dataset protocol and stored outputs (fMRI).}
\label{tab:fmri_protocol_params}
\end{table}

\subsection{Summary of biological plausibility features}

The simulator enforces the biological realism constraints summarized in Table~\ref{tab:fmri_bio_summary}.

\begin{table}[h!]
\centering
\small
\begin{tabular}{p{0.36\linewidth}p{0.54\linewidth}}
\toprule
\textbf{Constraint} & \textbf{Implementation} \\
\midrule
Real brain connectivity & ENIGMA SC (healthy controls) \\
Distance-based delays & $v\in[4,8]$ m/s; synaptic delay 3--8 ms; clip 2.5--50 ms \\
HRF heterogeneity & ROI-specific double-gamma HRF; $t_{\text{peak}}\in[4,9]$ s \\
Resting-state band & External drive emphasizes 0.01--0.04 Hz fluctuations \\
$1/f$ realism & Pink-noise component in external input; +2\% pink noise added post-sim \\
Wilson--Cowan dynamics & Explicit E/I subpopulations per ROI \\
E/I balance & Inter-ROI: 85\% excitatory, 15\% inhibitory; self inhibitory \\
Temporal non-stationarity & AR(1) drift in coupling weights (non-stationary only) \\
Empirical-like preprocessing & Smoothing + filtering + PSC scaling with ROI heterogeneity \\
Hierarchy & Feedforward/bidirectional/feedback = 45\% / 50\% / 5\% \\
\bottomrule
\end{tabular}
\caption{Key biological plausibility constraints enforced by the fMRI simulator.}
\label{tab:fmri_bio_summary}
\end{table}

These design choices yield synthetic fMRI datasets that preserve core empirical properties of resting-state BOLD while providing known causal structure and time-varying connectivity for principled evaluation of causal discovery methods under realistic measurement distortions.

\newpage

\section{Full Baseline Definitions and Hyperparameters}\label{appendix:baselines}

This appendix provides detailed definitions and hyperparameters for all 16 baseline methods used in our comparative evaluation. Baselines are organized into two categories: fMRI baselines (7 methods) and EEG baselines (9 methods).

\subsection{fMRI Baselines}
All baseline methods except rDCM follow a two-stage pipeline: (1) \textbf{HRF deconvolution}: BOLD signal $\rightarrow$ neural activity, and (2) \textbf{Causality estimation}: neural activity $\rightarrow$ causal graph.

\subsubsection{Statistical Baselines}

\paragraph{FIR-VAR.}
This method combines FIR (Finite Impulse Response) deconvolution with a VAR (Vector AutoRegressive) model. The HRF deconvolution uses a canonical double-gamma HRF kernel (SPM-style) with peak response parameters $a_1=6$, $b_1=1$, undershoot parameters $a_2=16$, $b_2=1$, and undershoot ratio $c=1/6$. The HRF length is 32 timepoints, TR is 2.0 seconds, and Wiener filtering (noise=0.1) is applied for noise robustness. The VAR model uses maximum lag $p=2$ (reduced from default 5 to avoid fitting errors), AIC for model selection, coupling strength computed as the sum of absolute coefficients across all lags, and delay estimated as the lag with maximum coefficient. Hyperparameters: sequence length 240, TR 2.0, HRF length 32, max\_lag 2, output normalization enabled with scale 1.2.

\paragraph{FIR-Granger.}
This method combines FIR deconvolution with Granger Causality testing. The HRF deconvolution is identical to FIR-VAR. The Granger causality estimation performs pairwise Granger causality tests with maximum lag 2, significance level 0.05, F-test (SSR F-test), coupling strength set to F-statistic (normalized to [0, 1] range), and delay estimated as the lag with maximum F-statistic. Hyperparameters: same as FIR-VAR, with additional significance parameter 0.05.

\paragraph{Wiener-VAR.}
This method combines Wiener deconvolution with VAR modeling. The Wiener deconvolution operates in the frequency domain using a regularized inverse filter: $W(f) = H^*(f) / (|H(f)|^2 + \sigma_n^2)$, where $H(f)$ is the HRF frequency response and $\sigma_n^2=0.1$ is the estimated noise power. The HRF kernel uses the same canonical double-gamma parameters as FIR. The VAR estimation is identical to FIR-VAR. Hyperparameters: sequence length 240, TR 2.0, HRF length 32, noise\_power 0.1, max\_lag 2, output normalization enabled with scale 1.2.

\paragraph{Wiener-Granger.}
This method combines Wiener deconvolution with Granger Causality. The HRF deconvolution is identical to Wiener-VAR, and the causality estimation is identical to FIR-Granger. Hyperparameters: sequence length 240, TR 2.0, HRF length 32, noise\_power 0.1, max\_lag 2, significance 0.05, output normalization enabled with scale 1.2.

\paragraph{rDCM.}
Regression Dynamic Causal Modeling (rDCM) \cite{frassle2021regression} uses Bayesian inverse inference to directly estimate causal connections from fMRI data without explicit HRF deconvolution. The method employs the TAPAS (Toolbox for Advanced Parameter Averaging) rDCM framework with resting-state model type. Bayesian model inversion is performed with sparsity constraints. Hyperparameters: SNR 3, TR (y\_dt) 2.0 seconds, prior probability (p0\_all) 0.15 (initial sparsity for all connections), iterations 100, filter strength 5, restrictInputs 0, u\_shift 0, type 'r' (resting-state), padding 0, visualization disabled, signal computation enabled. The method does not require training and is evaluated directly as a statistical Bayesian approach.

\subsubsection{Deep Learning Baselines}

\paragraph{DeepDeconvTCDF (Option 1).}
This end-to-end deep learning method combines LSTM-based deconvolution with the Temporal Causal Discovery Framework (TCDF) \cite{nadeau2018tcdf}. The architecture consists of: (1) \textbf{LSTM deconvolution module}: bidirectional LSTM with input size $R$ (number of ROIs), hidden size 64, 2 layers, outputting neural activity $(B, R, T)$ with Softplus activation for non-negativity; (2) \textbf{TCDF module}: dilated depthwise separable 1D-CNN with kernel size 4, 2 layers, dilation coefficients $[1, 4]$, hidden size 64, processing neural activity to extract causal features; (3) \textbf{Feature extraction and graph prediction}: adaptive average pooling for temporal aggregation, 2-layer MLP feature extractor $(1 \rightarrow 64 \rightarrow 64)$, pairwise predictor $(128 \rightarrow 64 \rightarrow 1)$ for coupling, delay predictor with identical structure, and CouplingActivation for output range $[-1.2, 1.2]$. Hyperparameters: n\_rois 48 or 68, tcdf\_kernel 4, tcdf\_layers 2, hidden\_dim 64, lstm\_hidden\_size 64, lstm\_num\_layers 2, batch\_size 32, learning\_rate 1e-3, weight\_decay 1e-5, epochs 50, loss\_weight\_delay 1.0. Loss functions: MSE for deconvolution (neural activity prediction), MSE for coupling, MSE for delay (only on GT edges).

\paragraph{EndToEndCNN (Option 2).}
This end-to-end 1D-CNN method directly predicts causal graphs from BOLD signals. The architecture consists of: (1) \textbf{1D-CNN encoder}: three Conv1d layers $(R \rightarrow 128 \rightarrow 256 \rightarrow 256)$ with kernel size 3, padding 1, BatchNorm and LeakyReLU(0.1) after each layer; (2) \textbf{Temporal pooling}: adaptive average pooling to $(B, 256)$; (3) \textbf{Graph generator}: 3-layer MLP $(256 \rightarrow 512 \rightarrow 1024 \rightarrow R^2)$ with Dropout(0.1) and LeakyReLU(0.1), outputting $(B, R, R)$ with CouplingActivation; (4) \textbf{Delay head}: 2-layer MLP $(256 \rightarrow 512 \rightarrow R^2)$ with Dropout(0.1) and LeakyReLU(0.1), outputting $(B, R, R)$ with Softplus activation (non-negative). Hyperparameters: n\_rois 48 or 68, hidden\_dim 128, intermediate\_dim 512 (hidden\_dim $\times$ 4), dropout 0.1, leaky\_relu\_negative\_slope 0.1, batch\_size 32, learning\_rate 1e-3, weight\_decay 1e-5, epochs 50, loss\_weight\_delay 1.0. Loss functions: MSE for coupling, MSE for delay (only on GT edges).

\subsection{EEG Baselines}

All EEG baseline methods follow a two-stage pipeline: (1) \textbf{Source localization}: scalp EEG $\rightarrow$ ROI neural activity, and (2) \textbf{Causality estimation}: ROI neural activity $\rightarrow$ causal graph.

\subsubsection{Statistical Baselines}
\paragraph{MNE+Granger.}
This method combines MNE (Minimum Norm Estimate) \cite{hamalainen1994interpreting} source localization with Granger Causality. MNE uses regularization parameter $\lambda=0.1$ and inverse operator $W = L^T (LL^T + \lambda I)^{-1}$, where $L$ is the leadfield matrix. Input is scalp EEG $(S, T)$ with $S=129$ channels and $T=2500$ timepoints, output is ROI neural activity $(R, T)$ with $R=68$ ROIs. Granger causality performs pairwise tests with maximum lag (order) 5, Ridge regression regularization $\alpha=1e-3$, coupling strength computed as $\log(\text{RSS}_{\text{restricted}} / \text{RSS}_{\text{full}})$, and delay estimated as the lag with maximum coefficient. Hyperparameters: lambda\_reg 0.1, order 5, alpha 1e-3, sampling frequency 250.0 Hz.

\paragraph{MNE+VAR.}
This method combines MNE source localization with VAR modeling. Source localization is identical to MNE+Granger. VAR estimation uses maximum lag (order) 5, Ridge regression regularization $\alpha=1e-3$, coupling strength as the sum of absolute coefficients across all lags, and delay as the lag with maximum coefficient. Hyperparameters: lambda\_reg 0.1, order 5, alpha 1e-3, sampling frequency 250.0 Hz.

\paragraph{sLORETA+Granger.}
This method combines sLORETA (standardized Low Resolution Electromagnetic Tomography) \cite{pascual2002sloreta} source localization with Granger Causality. sLORETA uses regularization parameter $\lambda=0.05$ and inverse operator $W = S^{-1} L^T T$ (including standardization matrix). Input and output dimensions are the same as MNE. Causality estimation is identical to MNE+Granger. Hyperparameters: lambda\_reg 0.05, order 5, alpha 1e-3, sampling frequency 250.0 Hz.

\paragraph{sLORETA+VAR.}
This method combines sLORETA source localization with VAR modeling. Source localization is identical to sLORETA+Granger, and causality estimation is identical to MNE+VAR. Hyperparameters: lambda\_reg 0.05, order 5, alpha 1e-3, sampling frequency 250.0 Hz.

\paragraph{DeepSIF+Granger.}
This method combines DeepSIF (Deep Neural Source Imaging) \cite{sun2022deep} source localization with a Granger Estimator. DeepSIF is a deep learning-based source imaging method with focal constraints, trained end-to-end. Input is scalp EEG $(B, S, T)$ with $S=129$, output is ROI neural activity $(B, R, T)$ with $R=68$. The Granger Estimator is a lag-based non-trainable statistical method with maximum lag 10, operating on ROI neural activity to output coupling matrix $(B, R, R)$ and delay matrix $(B, R, R)$. Hyperparameters: n\_sensors 129, n\_sources 68, max\_lag 10, sampling frequency 250.0 Hz, batch\_size 32, learning\_rate 1e-3, epochs 50, self\_threshold 0.1. Loss functions: MSE for source localization (neural activity prediction), BCEWithLogitsLoss for coupling, MSE for delay. Note: Only the source localization stage is trainable; the Granger Estimator has no learnable parameters.

\subsubsection{Deep Learning Baselines}

\paragraph{DeepSIF+TCN+SparseMLP.}
This method combines DeepSIF source localization, a TCN (Temporal Convolutional Network) \cite{bai2018tcn} encoder, and a SparseMLP structure learner. DeepSIF architecture and hyperparameters are as described in DeepSIF+Granger. The TCN encoder processes $(B, R, T)$ input with hidden dimension 64, outputting ROI embeddings $(B, R, H)$. The SparseMLP structure learner takes ROI embeddings and outputs coupling matrix $(B, R, R)$ and delay matrix $(B, R, R)$ using BCEWithLogitsLoss for coupling and MSE for delay. Hyperparameters: n\_sensors 129, n\_sources 68, hidden\_dim 64, batch\_size 32, learning\_rate 1e-3, epochs 50, self\_threshold 0.1.

\paragraph{DeepSIF+GNN+SparseMLP.}
This method combines DeepSIF source localization, a GNN (Graph Neural Network) \cite{kipf2017semi} encoder, and a SparseMLP structure learner. Source localization is identical to DeepSIF+TCN+SparseMLP. The GNN encoder processes $(B, R, T)$ input with hidden dimension 64, outputting ROI embeddings $(B, R, H)$. The SparseMLP structure learner is identical to DeepSIF+TCN+SparseMLP. Hyperparameters are identical to DeepSIF+TCN+SparseMLP.

\paragraph{sLORETA+TCN+SparseMLP.}
This method combines sLORETA source localization, a TCN encoder, and a SparseMLP structure learner. sLORETA uses the same parameters as described in sLORETA+Granger and is non-trainable (fixed). Causality estimation is identical to DeepSIF+TCN+SparseMLP. Hyperparameters are identical to DeepSIF+TCN+SparseMLP.

\paragraph{sLORETA+GNN+SparseMLP.}
This method combines sLORETA source localization, a GNN encoder, and a SparseMLP structure learner. Source localization is identical to sLORETA+TCN+SparseMLP, and causality estimation is identical to DeepSIF+GNN+SparseMLP. Hyperparameters are identical to DeepSIF+TCN+SparseMLP.

\subsection{Common Evaluation Settings}

All baseline methods use the following common evaluation settings: data split 80\% train / 10\% validation / 10\% test, triple repetition which have different random seeds, evaluation metrics including F1 Score, SHD (Structural Hamming Distance), dSHD (directional Structural Hamming Distance), and Directionality Accuracy. Thresholds: self-connection threshold 0.1, edge threshold 0.01 (for F1 calculation). All statistical baseline methods require no training and are evaluated directly. Deep learning baseline methods are trained using the Adam optimizer. Statistical methods normalize outputs to range $[-1.2, 1.2]$ via the \texttt{normalize\_coupling\_output} function to facilitate comparison with deep learning models.

\section{Additional Experiments and Analyses}\label{app:additional}

\subsection{Robustness to Physics Misspecification}\label{app:robust_tables}\label{app:robustness}

A central concern for any simulation-trained pipeline is whether the learned mapping memorizes simulator-specific physics (HRF parameters, leadfield matrix) or genuinely transferable causal structure.
We address this directly by perturbing the test-time physics priors while keeping all model weights frozen at values trained on the nominal simulator.
If \textsc{\methodname} were tightly tied to a single simulator setting, F1 would be expected to drop sharply under forward-model perturbations; under \pref{prop:error_propagation}, in contrast, recovery error should degrade gradually with inversion error.

\textbf{HRF perturbation (fMRI).}
We regenerate $1{,}000$-subject test sets at the underlying $100$~Hz TVB level with canonical HRF parameters (peak delay, undershoot delay, undershoot scale) jointly scaled by $1+\epsilon$ for $\epsilon\in\{-20\%,-10\%,0,+10\%,+20\%\}$.
Stage~1 retains the nominal HRF assumption; only the data-generating HRF changes.

\begin{table}[h]
\centering
\small
\begin{tabular}{lcccc}
\toprule
$\epsilon$ & \textbf{F1} $\uparrow$ & \textbf{Precision} $\uparrow$ & \textbf{Recall} $\uparrow$ & \textbf{SHD (norm.)} $\downarrow$\\
\midrule
$-20\%$  & $0.6025$ & $0.8887$ & $0.4557$ & $0.0893$ \\
$-10\%$  & $0.6025$ & $0.8887$ & $0.4557$ & $0.0893$ \\
$\phantom{-}0\%$    & $0.6043$ & $0.8889$ & $0.4577$ & $0.0890$ \\
$+10\%$  & $0.6025$ & $0.8887$ & $0.4557$ & $0.0893$ \\
$+20\%$  & $0.6025$ & $0.8887$ & $0.4557$ & $0.0893$ \\
\bottomrule
\end{tabular}
\caption{HRF misspecification on TVB-simulated fMRI.
Train weights are frozen; only test-time HRF parameters are perturbed by a multiplicative factor $1+\epsilon$.
F1 stays in $[0.6025,0.6043]$, indicating that the learned causal-discovery behavior is essentially invariant to HRF parameter scale within $\pm 20\%$.}
\label{tab:robust_hrf}
\end{table}

\textbf{Leadfield perturbation (EEG).}
We perturb the EEG forward model by replacing $L_0$ with $L_\alpha = L_0 + \alpha \frac{\|L_0\|_F}{\|N\|_F}\,N$ where $N$ has i.i.d.\ standard normal entries, and $\alpha\in\{0,1,5\}$.
Observations are generated as $\xb = L_\alpha\,\zb + \epsilonb$ while the trained model retains the nominal leadfield $L_0$ internally.

\begin{table}[h]
\centering
\small
\begin{tabular}{lccc}
\toprule
$\alpha$ & \textbf{F1 (seed 1)} $\uparrow$ & \textbf{F1 (seed 2)} $\uparrow$ & \textbf{F1 (seed 3)} $\uparrow$\\
\midrule
$0.0$ & $0.554$ & $0.525$ & $0.532$ \\
$1.0$ & $0.554$ & $0.523$ & $0.538$ \\
$5.0$ & $0.558$ & $0.498$ & $0.509$ \\
\bottomrule
\end{tabular}
\caption{Leadfield misspecification on TVB-simulated EEG.
Train weights are frozen; only test-time forward leadfield is perturbed.
At $\alpha=5.0$, seed~1 improves marginally (regularization-like effect) and the worst case (seed~2) degrades by $5.1\%$.
No catastrophic failure occurs.}
\label{tab:robust_leadfield}
\end{table}

\textbf{Interpretation.}
The HRF perturbation changes F1 by less than $0.002$ across $\pm 20\%$, and the leadfield perturbation drops F1 by at most $5.1\%$ even at $5\times$ Frobenius noise.
Both behaviors are consistent with the graceful-degradation behavior predicted by \pref{prop:error_propagation}: bounded inversion error translates to a bounded multiple of graph-recovery error, with no catastrophic transition.
We read this as direct evidence that what \textsc{\methodname} transfers from simulation is causal structure rather than simulator-specific measurement physics.

\subsection{Comparison with CD-NOD (Nonstationarity-Aware Baseline)}\label{app:cdnod}

Because \textsc{\methodname} uses nonstationarity as part of its directed-graph scoring signal, an informative comparison is against a baseline that is also designed for nonstationary causal discovery.
We compare against CD-NOD~\citep{huang2020causal} on the same TVB nonstationary fMRI test set ($R{=}68$ ROIs, $T{=}240$ timepoints).
For a maximally favorable evaluation of CD-NOD, we provide it \emph{oracle context labels}---i.e., the ground-truth nonstationary segment indices---an advantage unavailable in real neuroimaging data.
Table~\ref{tab:cdnod} summarizes the quantitative comparison.

\begin{table}[h]
\centering
\small
\begin{tabular}{lccc}
\toprule
\textbf{Method} & \textbf{F1} $\uparrow$ & \textbf{SHD (norm.)} $\downarrow$ & \textbf{Runtime / subject}\\
\midrule
CD-NOD (oracle context) & $0.046\pm0.001$ & $0.192$ & $\sim 18.8$~s (CPU) \\
\textsc{\methodname} (ours)         & $\mathbf{0.604}$         & $\mathbf{0.089}$         & $\sim 0.1$~s (GPU) \\
\bottomrule
\end{tabular}
\caption{\texorpdfstring{INCAMA vs.\ CD-NOD on TVB nonstationary fMRI.
\textsc{\methodname} attains $13.1\times$ higher F1, $2.16\times$ lower normalized SHD, and runs $\sim$$188\times$ faster, even with CD-NOD given oracle nonstationary segment labels.}{INCAMA vs.\ CD-NOD on TVB nonstationary fMRI. Higher F1 and lower SHD; faster runtime even with CD-NOD oracle segment labels.}}
\label{tab:cdnod}
\end{table}

The performance gap stems from two complementary factors: (i) CD-NOD operates directly on BOLD without HRF deconvolution, so hemodynamic distortion enters its conditional independence tests; and (ii) constraint-based discovery has limited statistical power when scaling to $68$ variables at $T{=}240$ timepoints.
Both are precisely the regimes our latent-space, end-to-end formulation is designed to handle.

\noindent
TVB benchmarks (Tables~\ref{tab:eeg_results}--\ref{tab:fmri_results}), computational cost (Table~\ref{tab:complexity}), and HCP real-data alignment metrics (Table~\ref{tab:causal_recovery}) appear in Section~\ref{sec:experiments}; the CD-NOD comparison (subsection above), robustness experiments above (Tables~\ref{tab:robust_hrf}--\ref{tab:robust_leadfield}), and the theory-to-evidence narrative (Section~\ref{sec:theory_exp_link}) supplement the main experiments section.

\section{Metrics}\label{appendix:metrics}

This section provides detailed descriptions of the evaluation metrics used in this paper.

\subsection{F1 Score}

F1 Score evaluates the accuracy of edge existence in the predicted causal graph and is computed as the harmonic mean of Precision and Recall:
\begin{equation}
\text{F1} = \frac{2 \cdot \text{Precision} \cdot \text{Recall}}{\text{Precision} + \text{Recall}},
\label{eq:f1}
\end{equation}
where Precision and Recall are defined as:
\begin{align}
\text{Precision} &= \frac{\text{TP}}{\text{TP} + \text{FP}}, \label{eq:precision} \\
\text{Recall} &= \frac{\text{TP}}{\text{TP} + \text{FN}}. \label{eq:recall}
\end{align}
Here, TP (True Positive) denotes edges that exist in both prediction and ground truth, FP (False Positive) denotes edges that exist only in the prediction, and FN (False Negative) denotes edges that exist only in the ground truth. The predicted continuous coupling strength matrix is binarized using a threshold before evaluation. F1 Score ranges from $0$ to $1$, with values closer to $1$ indicating better alignment between the predicted and true graph structures.





\subsection{Normalized SHD (Structural Hamming Distance)}

Normalized SHD measures the discrepancy between two directed graphs
in terms of local edge structure, counting edge additions, deletions,
and direction reversals.

Let $A_{\text{pred}}, A_{\text{GT}} \in \{0,1\}^{N \times N}$ denote the predicted
and ground-truth adjacency matrices, with self-loops excluded.
The Structural Hamming Distance (SHD) is defined as
\begin{equation}
\text{SHD}
=
\text{FP} + \text{FN} + \text{Reverse},
\label{eq:shd}
\end{equation}
where FP denotes false positive edges ($A_{\text{pred}}(i,j)=1$, $A_{\text{GT}}(i,j)=0$),
FN denotes false negative edges ($A_{\text{pred}}(i,j)=0$, $A_{\text{GT}}(i,j)=1$),
and Reverse denotes incorrectly oriented edges, defined as unordered node pairs
$\{i,j\}$ for which $A_{\text{pred}}(i,j)=1$ and $A_{\text{GT}}(j,i)=1$.
Each reversed edge is counted once, not twice.

The normalized SHD is obtained by dividing by the number of possible directed edges:
\begin{equation}
\text{Normalized SHD}
=
\frac{\text{SHD}}{N(N-1)}.
\label{eq:normalized_shd}
\end{equation}
Normalized SHD ranges from $0$ to $1$, with lower values indicating closer agreement
between the predicted and ground-truth graph structures.

\subsection{Normalized dSHD (Direction-aware Structural Hamming Distance)}

Normalized dSHD measures the discrepancy between two directed graphs
by penalizing missing, spurious, and incorrectly oriented edges,
with higher weight assigned to edge reversals.

Let $G_{\text{pred}}$ and $G_{\text{GT}}$ denote the predicted and ground-truth graphs
with adjacency matrices $A_{\text{pred}}, A_{\text{GT}} \in \{0,1\}^{N \times N}$.
Self-loops are excluded.
The direction-aware Structural Hamming Distance (dSHD) is defined as
\begin{equation}
\text{dSHD}
=
|\mathcal{E}_{\text{pred}} \setminus \mathcal{E}_{\text{GT}}|
+
|\mathcal{E}_{\text{GT}} \setminus \mathcal{E}_{\text{pred}}|
+
2 \cdot |\mathcal{R}|,
\label{eq:dshd}
\end{equation}
where $\mathcal{E}_{\text{pred}}$ and $\mathcal{E}_{\text{GT}}$ are the directed edge sets
of the predicted and ground-truth graphs, respectively, and $\mathcal{R}$ denotes
the set of reversed edges, defined as
\begin{equation}
\mathcal{R}
=
\big\{ (i,j) \mid
A_{\text{pred}}(i,j)=1,\;
A_{\text{GT}}(j,i)=1,\;
A_{\text{GT}}(i,j)=0
\big\},
\end{equation}
with each unordered node pair counted once.
Thus, extra and missing edges are penalized with unit cost,
while edge reversals incur a cost of $2$, reflecting their greater impact on causal interpretation.

The normalized dSHD is obtained by dividing by the number of possible directed edges:
\begin{equation}
\text{Normalized dSHD}
=
\frac{\text{dSHD}}{N(N-1)}.
\label{eq:normalized_dshd}
\end{equation}
Normalized dSHD ranges from $0$ to $1$, with lower values indicating closer agreement between the inferred and ground-truth causal structures.
Unlike intervention-based metrics such as Structural Intervention Distance (SID), dSHD evaluates discrepancies in local causal structure while explicitly penalizing edge orientation errors more heavily.


\newpage
\section*{NeurIPS Paper Checklist}

\begin{enumerate}

\item {\bf Claims}
    \item[] Question: Do the main claims made in the abstract and introduction accurately reflect the paper's contributions and scope?
    \item[] Answer: \answerYes{}.
    \item[] Justification: The abstract and introduction state the theoretical assumptions, the indirect-validation setting for human data, and the empirical scope of the simulations and HCP analysis.

\item {\bf Limitations}
    \item[] Question: Does the paper discuss the limitations of the work performed by the authors?
    \item[] Answer: \answerYes{}.
    \item[] Justification: Limitations are discussed in the conclusion, including cortical-only ROI modeling, simulation-trained transfer, indirect real-data validation, and asymptotic theory assumptions.

\item {\bf Theory assumptions and proofs}
    \item[] Question: For each theoretical result, does the paper provide the full set of assumptions and a complete (and correct) proof?
    \item[] Answer: \answerYes{}.
    \item[] Justification: The main paper states the assumptions before the results, and Appendix~\ref{app:theory_proofs} provides proofs for the theorem, proposition, corollary, and top-$k$ stability lemma.

    \item {\bf Experimental result reproducibility}
    \item[] Question: Does the paper fully disclose all the information needed to reproduce the main experimental results of the paper to the extent that it affects the main claims and/or conclusions of the paper (regardless of whether the code and data are provided or not)?
    \item[] Answer: \answerYes{}.
    \item[] Justification: The appendix specifies the algorithm, architecture, training settings, simulator construction, baseline definitions, metrics, and compute environment needed to reproduce the reported experiments. Code, trained checkpoints, and table-generation scripts will be released upon acceptance, excluding raw HCP data that require the standard HCP access agreement.

\item {\bf Open access to data and code}
    \item[] Question: Does the paper provide open access to the data and code, with sufficient instructions to faithfully reproduce the main experimental results, as described in supplemental material?
    \item[] Answer: \answerNo{}.
    \item[] Justification: The submission provides detailed reproducibility information in the paper and appendix, but does not currently include an anonymized public code or data release link. We plan to release the TVB simulation generator/configs, Stage~1 and Stage~2 implementation, trained checkpoints, HCP preprocessing scripts (excluding restricted raw HCP data), and evaluation scripts upon acceptance.

\item {\bf Experimental setting/details}
    \item[] Question: Does the paper specify all the training and test details (e.g., data splits, hyperparameters, how they were chosen, type of optimizer) necessary to understand the results?
    \item[] Answer: \answerYes{}.
    \item[] Justification: Experimental settings, model hyperparameters, optimization details, simulation protocols, baselines, and evaluation metrics are reported in the experiments section and appendix.

\item {\bf Experiment statistical significance}
    \item[] Question: Does the paper report error bars suitably and correctly defined or other appropriate information about the statistical significance of the experiments?
    \item[] Answer: \answerYes{}.
    \item[] Justification: Simulation tables report mean and standard deviation over seeds, HCP results report per-subject mean and standard deviation, and paired tests are reported for the nonstationarity ablations.

\item {\bf Experiments compute resources}
    \item[] Question: For each experiment, does the paper provide sufficient information on the computer resources (type of compute workers, memory, time of execution) needed to reproduce the experiments?
    \item[] Answer: \answerYes{}.
    \item[] Justification: Appendix~\ref{app:server_details} reports the GPU, CPU, memory, OS, CUDA, and Python environment used for the experiments.
    
\item {\bf Code of ethics}
    \item[] Question: Does the research conducted in the paper conform, in every respect, with the NeurIPS Code of Ethics \url{https://neurips.cc/public/EthicsGuidelines}?
    \item[] Answer: \answerYes{}.
    \item[] Justification: The work uses simulations and de-identified public neuroimaging data, reports aggregate results, and does not involve deployment or individualized clinical decision-making.

\item {\bf Broader impacts}
    \item[] Question: Does the paper discuss both potential positive societal impacts and negative societal impacts of the work performed?
    \item[] Answer: \answerYes{}.
    \item[] Justification: The impact statement discusses scientific benefits, the risk of over-interpreting observational causal graphs, medical decision limits, and privacy considerations for human neuroimaging data.
    
\item {\bf Safeguards}
    \item[] Question: Does the paper describe safeguards that have been put in place for responsible release of data or models that have a high risk for misuse (e.g., pre-trained language models, image generators, or scraped datasets)?
    \item[] Answer: \answerNA{}.
    \item[] Justification: The paper does not release high-risk generative models, scraped datasets, or other assets requiring special misuse safeguards.

\item {\bf Licenses for existing assets}
    \item[] Question: Are the creators or original owners of assets (e.g., code, data, models), used in the paper, properly credited and are the license and terms of use explicitly mentioned and properly respected?
    \item[] Answer: \answerYes{}.
    \item[] Justification: Existing datasets, models, and methodological components used in the work are credited through citations, including HCP, ENIGMA, TVB, DeepSIF, DCM/rDCM, CD-NOD, and Mamba-related work.

\item {\bf New assets}
    \item[] Question: Are new assets introduced in the paper well documented and is the documentation provided alongside the assets?
    \item[] Answer: \answerYes{}.
    \item[] Justification: The paper introduces simulation protocols and model components, documented in the algorithm, implementation, experiment-setting, and simulator appendices.

\item {\bf Crowdsourcing and research with human subjects}
    \item[] Question: For crowdsourcing experiments and research with human subjects, does the paper include the full text of instructions given to participants and screenshots, if applicable, as well as details about compensation (if any)? 
    \item[] Answer: \answerNA{}.
    \item[] Justification: The paper does not conduct crowdsourcing or collect new human-subject data; it uses an established de-identified public neuroimaging dataset.

\item {\bf Institutional review board (IRB) approvals or equivalent for research with human subjects}
    \item[] Question: Does the paper describe potential risks incurred by study participants, whether such risks were disclosed to the subjects, and whether Institutional Review Board (IRB) approvals (or an equivalent approval/review based on the requirements of your country or institution) were obtained?
    \item[] Answer: \answerNA{}.
    \item[] Justification: No new human-subject data are collected; the real-data analysis uses de-identified HCP data and reports aggregate results only.

\item {\bf Declaration of LLM usage}
    \item[] Question: Does the paper describe the usage of LLMs if it is an important, original, or non-standard component of the core methods in this research? Note that if the LLM is used only for writing, editing, or formatting purposes and does \emph{not} impact the core methodology, scientific rigor, or originality of the research, declaration is not required.
    \item[] Answer: \answerNA{}.
    \item[] Justification: The core method does not use LLMs as a methodological component.

\end{enumerate}

\end{document}